\documentclass[trans]{IEEEtran}

\usepackage{xspace,amsmath,amssymb,amsfonts,epsfig,syntonly,amsthm}
\usepackage{cite,bm,color,url,textcomp}
\usepackage{algorithmicx,algorithm,algpseudocode}
\usepackage{epstopdf}
\usepackage{empheq}
\usepackage{balance}
\usepackage{graphicx}

\newtheorem{lemma}{\hspace{-0pt}\bf Lemma}

\newtheorem{theorem}{\hspace{-0pt}\bf Theorem}

\newtheorem{remark}{\hspace{-0pt}\bf Remark}

\DeclareMathOperator*{\argmin}{arg\,min}

\long\def\symbolfootnote[#1]#2{\begingroup
\def\thefootnote{\fnsymbol{footnote}}
\footnote[#1]{#2}\endgroup}

\psfull

\allowdisplaybreaks[2]

\begin{document}
\title{Multi-Timescale Online Optimization of Network Function Virtualization for Service Chaining}

\author{Xiaojing Chen, Wei Ni, Tianyi Chen, Iain B. Collings,~\IEEEmembership{Fellow, IEEE}, Xin Wang,~\IEEEmembership{Senior Member, IEEE}, \\Ren Ping Liu,~\IEEEmembership{Senior Member, IEEE}, and Georgios B. Giannakis,~\IEEEmembership{Fellow, IEEE}
{\thanks{Work in this paper was supported by the National Natural Science Foundation of China grant 61671154, the National Key Research and Development Program of China grant 2017YFB04034002, and the Innovation Program of Shanghai Municipal Science and Technology Commission grant 17510710400; US NSF grants 1509005, 1508993, 1423316, and 1442686.
}
\thanks{X. Chen and X. Wang are with the Shanghai Institute for Advanced Communication and Data Science, Key Laboratory for Information Science of Electromagnetic Waves (MoE), the Dept. of Communication Science and Engineering, Fudan University, 220 Handan Road, Shanghai, China. Emails:~\{13210720095, xwang11\}@fudan.edu.cn.; X. Chen is also with the School of Engineering, Macquarie University, Sydney, NSW 2109, Australia.}
\thanks{W. Ni is with the Commonwealth Scientific and Industrial Research Organization (CSIRO), Sydney, NSW 2122, Australia. Email:~wei.ni@data61.csiro.au.}
\thanks{T. Chen and G. B. Giannakis are with the Dept. of Electrical and Computer Engineering and the Digital Technology Center, University of Minnesota, Minneapolis, MN 55455 USA. Emails:~\{chen3827, georgios\}@umn.edu.}
\thanks{I. B. Collings is with the School of Engineering, Macquarie University, Sydney, NSW 2109, Australia. Email:~iain.collings@mq.edu.au.}
\thanks{R. P. Liu is with the School of Electrical and Data Engineering, University of Technology Sydney, Sydney, NSW 2007, Australia. Email:~RenPing.Liu@uts.edu.au}
\thanks{Part of this paper has been presented in \cite{xiaojingglo17} without detailed proofs and analyses. Apart from providing the details, this paper has significant extensions on the placement of VNFs at different timescales and general application scenarios where there can be multiple VNFs installed per VM.}
	
}
}
\maketitle

\begin{abstract}
Network Function Virtualization (NFV) can cost-efficiently provide network services by running different virtual network functions (VNFs) at different virtual machines (VMs) in a correct order. This can result in strong couplings between the decisions of the VMs on the placement and operations of VNFs. This paper presents a new fully decentralized online approach for optimal placement and operations of VNFs. Building on a new stochastic dual gradient method, our approach decouples the real-time decisions of VMs, asymptotically minimizes the time-average cost of NFV, and stabilizes the backlogs of network services with a cost-backlog tradeoff of $[\epsilon,1/\epsilon]$, for any $\epsilon > 0$. Our approach can be relaxed into multiple timescales to have VNFs (re)placed at a larger timescale and hence alleviate service interruptions. While proved to preserve the asymptotic optimality, the larger timescale can slow down the optimal placement of VNFs. A learn-and-adapt strategy is further designed to speed the placement up with an improved tradeoff $[\epsilon,\log^2(\epsilon)/{\sqrt{\epsilon}}]$. Numerical results show that the proposed method is able to reduce the time-average cost of NFV by 30\% and reduce the queue length (or delay) by 83\%, as compared to existing benchmarks.
\end{abstract}

\section{Introduction}

Decoupling dedicated hardware from network services and replacing them with programmable virtual machines (VMs), Network Function Virtualization (NFV) is able to provide critical network functions on top of optimally shared physical infrastructure~\cite{li2015software, han2015network}. This can avoid disproportional hardware investments on short-lived functions, and adapt quickly as network functions evolve~\cite{mijumbi2016network}. Particularly, a virtual network function (VNF) is a virtualized task formerly carried out by proprietary and dedicated hardware, which moves network functions out of dedicated hardware devices and into software~\cite{mijumbi2016network}.

A network service (or service chain) can consist of multiple VNFs, which need to be run in a predefined order at different VMs running different VNF instances (i.e., software)~\cite{eramo2017approach}. Challenges arise from making optimal online decisions on the placement of VNFs, and the processing and routing of network services at each VM, especially in large-scale network platforms. On one hand, given the sequence of VNFs to be executed per network service, the optimal decisions of individual VMs are coupled. On the other hand, stochasticity prevails in the arrivals of network services, and the link capacity between VMs stemming from concurrent traffic~\cite{Riggio2016Scheduling}. Prices can also vary for the service of a VM, depending on the pricing policy of the service providers. The possibility of leveraging temporal resource variability implies the couplings of optimal decisions over time~\cite{mashayekhy2016online,chen2016cost}. Other challenges also include limited scalability resulting from centralized designs~\cite{qu2016delay}. 

These are open problems and have not been captured in previous works on VNF placement. The work in \cite{yousaf2015cost} focused on the placement of VNFs under the assumption of persistent arrivals of network services, where network services were instantly processed at the VMs admitting them and network service chains cannot be supported. The work in \cite{Cohen2015Near} addressed the placement of VNFs in a capacitated cloud network. The placement problem was formulated as a generalization of the Facility Location and Generalized Assignment problem; near-optimal solutions were provided with bi-criteria constant approximations. However, the model in \cite{Cohen2015Near} cannot account for function ordering or flow routing optimization.

Taking network service chains into account, recent works studied optimal decision-makings on processing and routing network services, under the assumption of persistent service arrivals~\cite{Addis2015Virtual}. NP-complete mixed integer linear programming (MILP) was formulated to minimize the delay of network service chains~\cite{qu2016delay}. A heuristic genetic approach was developed to solve MILP by sacrificing optimality~\cite{qu2016delay}. Greedy algorithms were developed to minimize flowtime or cost, or maximize revenue at a snapshot
of the network~\cite{mijumbi2015design}. These heuristic methods need to run in a centralized manner, thereby limiting scalability. Moreover, none of them have taken random service arrivals or dynamic pricing into account.

%

In this paper, we propose a new approach to distributed online optimization of NFV, where asymptotically optimal decisions on the placement and operation of NFV are spontaneously generated at individual VMs. Accounting for random service arrivals and time-varying prices, a stochastic dual gradient method is developed to decouple optimal decision-makings across different VMs and different time slots. It minimizes the time-average cost of NFV while stabilizing the queues of VMs. The gradients can be interpreted as the backlogs of the queues at every VM, and updated locally by the VM. With a proved cost-delay tradeoff $[\epsilon,{1}/{\epsilon}]$, the proposed method is able to asymptotically approach the global optimum which would be generated offline at a prohibitive complexity, by tuning the coefficient $\epsilon$.

Another important contribution is that we extend the distributed online optimization of NFV to two timescales, where the placement of VNFs is carried out at the VMs at a much larger time interval than the operations of the VNFs. This can effectively alleviate the interruptions that the placement/installation of VNFs can cause to network service provisions. We prove that the optimality loss of the two-timescale placement and operation of NFV is upper bounded, and the asymptotic optimality of the proposed distributed online optimization is preserved.
Other contribution is that we further speed up the placement of VNFs and improve the cost-delay tradeoff to $[\epsilon,\log^2(\epsilon)/{\sqrt{\epsilon}}]$ by designing a learn-and-adapt approach, where the statistics of network dynamics is learned from history with increasing accuracy. Corroborated by numerical results, the proposed approach is able to reduce the time-average cost of NFV by 30\% and reduce the queue lengths (or delays) by 83\%, as compared to existing non-stochastic approaches. Accounting for service chaining in stochastic NFV scenarios, the proposed algorithm is important and practical.

In a different yet relevant context, stochastic optimization has been developed for single- or multi-timescale resource allocation \cite{xiaojing2017, wang2016dynamic, wang2016two, wang2016, Yao2013Power,Sun2016Distributed,Alihemmati2017Multi}, routing \cite{Neely2009Optimal, Huang2010The}, and service computing \cite{Huang2012A} in queueing systems to deal with stochastic arrivals of workloads or energy. In particular, backpressure routing algorithms were developed in \cite{Neely2009Optimal, Huang2010The} to maximize throughput or minimize time-average costs in distributed (wireless) mesh networks, while stabilizing the queues across the networks. Yet, the backpressure algorithms were only focused on the routing of workloads, and did not involve any workload processing.
An energy-efficient offloading algorithm was proposed for mobile computing in \cite{Huang2012A}, which can dynamically offload part of an application's computation request to a dedicated server.
None of the existing approaches \cite{xiaojing2017, wang2016dynamic, wang2016two, wang2016, Yao2013Power,Sun2016Distributed,Alihemmati2017Multi, Huang2012A,Neely2009Optimal, Huang2010The} have taken sequentially chained network services into account.
Distinctively different from the existing methods, our approach is able to process and offload/forward sequentially chained network services which strongly couple the optimal decisions of VMs on routing and processing in time and space (i.e., among VMs). In particular, our approach decouples the strong coupling, optimizes both the processing and offloading of chained network services, and substantially reduces the queue lengths (or delays) by taking a learn-and-adapt strategy.

The rest of this paper is organized as follows. In Section~II, the system model is described. In Section~III, the distributed online optimization of the placement of VNFs and the processing/routing of network services is developed. Optimal placement and operation of NFV are investigated at two different timescales in Section~IV to prevent its interruptions to network service provisions. Numerical tests are provided in Section V, followed by concluding remarks in Section VI. Notations in the paper are listed in Table I.

\section{System Model}

\begin{figure}
\centering
\includegraphics[height=0.325\textwidth]{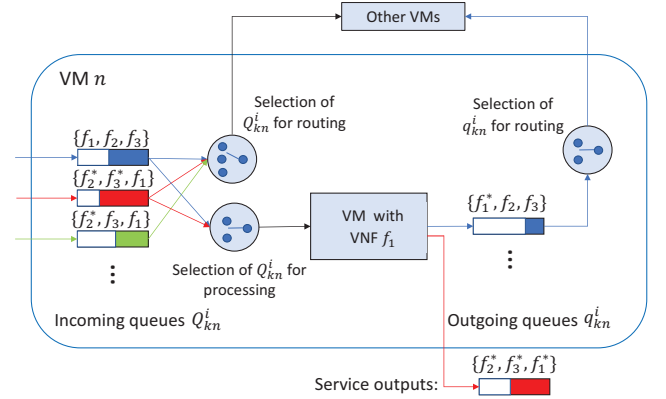}
\caption{An illustration on the processing and routing procedure of network services within VM $n$, where $f^*$ indicates processed VNFs.}\label{model}
\end{figure}

\begin{table}
  \caption{Variables}
  \label{tab:parameters}
  \center
  \begin{tabular}{p{4cm}|p{4cm}}
    \hline
    \emph{Basic Variables} \\
    \hline\hline
    $Q_{kn}^i$ & Queue length of type-$i$ network services to be processed by VNF $f_k$ at VM $n$\\
    \hline
   $q_{kn}^i$ & Queue length of type-$i$ network services after processed at VM $n$ and to be further processed by VNF $f_k$ at downstream VMs\\
    \hline
  $R_{kn}^{i,t}$ & Arrival rate of new type-$i$ network services to be processed by VNF $f_k$ at VM $n$\\
    \hline
    $\alpha_{n}^t$ & Time-varying price for service processing at VM $n$\\
    \hline
   $\beta_{[a,b]}^t$  & Time-varying price for service routing over link $[a,b]$\\
    \hline
     \emph{Decision Variables} \\
    \hline\hline
    $e_{kn}$ & Ability of VM $n$ to process $f_k$\\
    \hline
    $u_{k,ab}^i$ & Transmit rate of $Q_{ka}^{i}$ over link $[a,b]$\\
    \hline
 $v_{k,ab}^i$ & Transmit rate of $q_{ka}^{i}$ over link $[a,b]$\\
    \hline
   $p_{kn}^i$  & Processing rate of VM $n$ for type-$i$ network services to be processed by VNF $f_k$\\
    \hline

  \end{tabular}
\end{table}

Consider a platform consisting of $N$ VMs, supporting $K$ VNFs, and operating in a (possibly infinite) scheduling horizon consisting
of $T$ slots with a normalized slot duration ``1''.
The VMs can be located separately at different host servers, or co-located at the same host server. Assume that every VM can admit network services, and output the results; see Fig.~\ref{model}. This is consistent with existing designs of NFV systems.
Let $\mathcal{N}=\{1,\ldots,N\}$ collect the $N$ VMs.
Let $f_k$ ($k=1,\ldots,K$) denote the $k$-th VNF which can only be processed at the VMs running the corresponding software.

Assume that every VM runs a single VNF. (The extended scenario with a VM running multiple VNFs will be discussed in Section V-B). Let $e_{kn}(t) \in \{0,1\}, \forall k,n,t,$ denote the ability of VM $n$ to process $f_k$ at time $t$. Let $e_{kn}(t)=1$, if VM $n$ is installed with VNF $f_k$; and $e_{kn}(t)=0$, otherwise. We have $\sum_k e_{kn}(t)=1, \forall n,t$. 

Let $\mathcal{I}$ collect all possible types of network services, and each type of network service is to be processed by a permuted sequence of $\{f_1,\ldots,f_K\}$ or its subset.
A network service needs to traverse among multiple VMs, until the service are processed by the related VNFs in the correct order.

We design up to $2|\mathcal{I}|K$ service queues at each VM $n$. (For analytic tractability, we assume here that the VMs have sufficiently large memories, and the queues do not overflow. Nevertheless, one of the constraints we consider in this paper is that the time-average lengths of the queues are finite. As will be proved in Section III-A, the time-average lengths of the queues can be adjusted and reduced through parameter reconfiguration.) Half of the queues buffers network services to be processed by $f_k$ ($k=1,\ldots,K$); and the other half buffers the results of the first half after processed by $f_k$ and to be routed to downstream VMs for further processing.
Let $Q_{kn}^{i}(t)$ denote the queue lengths of type-$i$ network services to be processed by VNF $f_k$ at VM $n$ per slot $t$. Let $q_{kn}^{i}(t)$ denote the queue lengths of type-$i$ network services after processed by VNF $f_{k'}$ at VM $n$, and to be processed by VNF $f_k$ at downstream VMs per slot $t$. Here, $f_{k'}$ denotes the VNF which needs to be run prior to $f_k$ for type-$i$ network services; $\mathbf{Q}(t)=\{Q_{kn}^{i}(t), \forall n \in \mathcal{N}, i \in \mathcal{I}, k=1,\ldots,K\}$, $\mathbf{q}(t)=\{q_{kn}^{i}(t), \forall n \in \mathcal{N}, i \in \mathcal{I}, k=1,\ldots,K\}$, and  $\mathbf{A}(t)=\{\mathbf{Q}(t),\mathbf{q}(t)\}$. 
Let $R_{kn}^{i,t}\leq R^{\max}$ denote the arrival rate (in the number of services per slot) of new type-$i$ network services to be processed by VNF $f_k$ at VM $n$, where $R^{\max}$ is the maximum arrival rate. 

We assume that any VM $n$ or directional link $[a,b] (\forall a, b \in \mathcal{N})$ can only process or transmit a single type of network service per slot.
Let $u_{k,ab}^{i}(t)$ denote the transmit rate (in the number of services per slot) of $Q_{ka}^{i}(t)$, over directional link $[a,b]$ at time slot $t$. Let $v_{k,ab}^{i}(t)$ denote the transmit rate of $q_{ka}^{i}(t)$, over directional link $[a,b]$ at time slot $t$. It is easy to see that at any time $t$, we have
\begin{align}\label{eq: line rate}
&u_{k,ab}^{i}(t) \geq 0,\;\; v_{k,ab}^{i}(t) \geq 0, \;\;\forall i, k, [a,b],\notag\\
&\sum_{i,k} [u_{k,ab}^{i}(t)+v_{k,ab}^{i}(t)]=u_{k^*,ab}^{i^*}(t) (\text{or} \;v_{k^*,ab}^{i^*}(t))\leq l_{ab}^{\max},
\end{align}
where $l_{ab}^{\max}$ is the maximum transmit rate of link $[a,b]$, and a type-$i$ service which is to be processed by VNF $f_{k^*}$, is selected to be forwarded. 

Let $p_{kn}^{i}(t)$ denote the processing rate of VM $n$ (in the number of services per slot) for the type-$i$ network services to be processed by VNF $f_{k}$ at time slot $t$. We have
\begin{equation}\label{eq: service rate}
p_{kn}^{i}(t) \geq 0, \;\; \sum_{i,k} p_{kn}^{i}(t)e_{kn}(t) =p_{k^*n}^{i^*}(t)\leq p_{n}^{\max},\; \forall i, k, n,
\end{equation}
where $p_{n}^{\max}$ is the maximum processing rate of VM $n$, and a type-$i$ service which is to be processed by $f_{k^*}$, is selected to be processed.

Therefore, the queue length of type-$i$ network services to be processed by VNF $f_{k}$ at VM $n$ follows, $\forall i,k,n,t,$
\begin{equation}\label{eq: Qt}
\begin{aligned}
&Q_{kn}^{i}(t+1)=\max\{Q_{kn}^{i}(t)-\sum_{b\in \mathcal{N}} u_{k,nb}^{i}(t)-p_{kn}^{i}(t)e_{kn}(t), 0\}\\
&+\sum_{a \in \mathcal{N}} u_{k,an}^{i}(t)+\sum_{c \in \mathcal{N}} v_{k,cn}^{i}(t)+R_{kn}^{i,t}.
\end{aligned}
\end{equation}
The queue length of type-$i$ network services after processed by VNF $f_{k'}$ at VM $n$, and to be processed by VNF $f_k$ at downstream VMs follows, $\forall i,k,n,t,$
\begin{equation}\label{eq: qt}
\begin{aligned}
&q_{kn}^{i}(t+1)=\max\{q_{kn}^{i}(t)-\sum_{d\in \mathcal{N}} v_{k,nd}^{i}(t), 0\}+p_{k'n}^{i}(t)e_{k'n}(t),
\end{aligned}
\end{equation}
where $q_{kn}^{i}(t)=0$ ($k\in \emptyset$) in the case that the type-$i$ network services complete processing at the terminal VM $n$ and are output from the platform.

We define the network is \textit{stable} if and only if the following is met \cite{neely2010stochastic}:
\begin{align}\label{eq: stable}
\lim_{T\rightarrow \infty}\frac{1}{T}\sum_{t=1}^T  \mathbb{E}[|\mathbf{A}(t)|]\leq \infty.
\end{align}



Considering the processing and routing cost of network services in the platform, we define the total cost of routing services over all links and running VNFs on all VMs per slot $t$, as given by:
\begin{equation}\label{eq: average_cost}
\Phi^t(\mathbf{e}^t,\mathbf{p}^t,\mathbf{u}^t,\mathbf{v}^t):= \phi_1^t(\mathbf{e}^t,\mathbf{p}^t)+\phi_2^t(\mathbf{u}^t,\mathbf{v}^t),
\end{equation}
where $\mathbf{e}^t=\{e_{kn}(t), \forall k, n\}$, $\mathbf{p}^t=\{p_{kn}^{i}(t), \forall i, k, n\}$, $\mathbf{u}^t=\{u_{k,ab}^{i}(t), \forall i, k, [a,b]\}$, and $\mathbf{v}^t=\{v_{k,ab}^{i}(t), \forall i, k, [a,b]\}$. In \eqref{eq: average_cost}, $\phi_1^t(\mathbf{e}^t,\mathbf{p}^t)=\sum_{i,k,n}\alpha_n^t (p_{kn}^{i}(t)e_{kn}(t))^2$ and $\phi_2^t(\mathbf{u}^t,\mathbf{v}^t)=\sum_{a,b,i,k}\beta_{[a,b]}^t [(u_{k,ab}^{i}(t))^2+(v_{k,ab}^{i}(t))^2]$ are the costs that the network service provider charges for usages of VMs and links, respectively, e.g., following a quadratic pricing policy~\cite{zhu2008nonlinear,chen2017learn}. For non-elastic network services, e.g., video streaming and multimedia applications, these with urgent demand for high resources (i.e., bandwidths and CPUs) would charge for high prices~\cite{zhu2008nonlinear}. Here, $\alpha_n^t$ is the time-varying price for service processing  at VM~$n$ and $\beta_{[a,b]}^t$ is the time-varying price for service routing over link $[a,b]$.

\section{Distributed Online Optimization of Placement and Operation of NFV}

The objective of this paper is to minimize the time-average cost of NFV on a network platform while preserving the stability of the platform [cf. \eqref{eq: stable}], under random network service arrivals and prices. This can be achieved by making stochastically optimal decisions on processing or routing network services at every VM and time slot in a distributed fashion.
Let $\mathbf{x}^t:=\{\mathbf{e}^t, \mathbf{p}^t,\mathbf{u}^t,\mathbf{v}^t\}$ and $\mathcal{X}:=\{\mathbf{x}^t,\forall t\}$. The problem of interest is to solve
\begin{equation}\label{eq: objective}
\begin{aligned}
\Phi^* &= \min_{\mathcal{X}} \lim_{T\rightarrow \infty}\frac{1}{T}\sum_{t=1}^T \mathbb{E}\{\Phi^t(\mathbf{x}^t)\} \\
&\text{s.t.} \quad \eqref{eq: line rate},\; \eqref{eq: service rate},\;\eqref{eq: Qt},\; \eqref{eq: qt},\; \eqref{eq: stable}, \;\forall t
\end{aligned}
\end{equation}
where the expectation of $\Phi^t(\mathbf{x}^t)$ is taken over all randomnesses. The service arrival rate $\{R_{kn}^{i,t}, \forall i,k,n,t\}$ and the routing and processing prices $\{\alpha_n^t,\beta_{[a,b]}^t, \forall [a,b],n,t\}$ are all random.

\subsection{Dual gradient and asymptotic optimality}

It is difficult to solve \eqref{eq: objective} since we aim to minimize the average cost over an infinite time horizon. In particular, the queue dynamics in \eqref{eq: Qt} and \eqref{eq: qt} couple the optimization variables in time, rendering intractability for traditional solvers. Combining \eqref{eq: Qt} and \eqref{eq: qt} with \eqref{eq: stable}, however, it can be shown that in the long term, the service processing and routing rates must satisfy the following necessary conditions of queue stability \cite{neely2010stochastic}
\begin{subequations}\label{eq: Queue-relax}
\begin{align}
\lim_{T\rightarrow \infty}\frac{1}{T}\sum_{t=1}^T \mathbb{E} &[\sum_{a \in \mathcal{N}} u_{k,an}^{i}(t)+\sum_{c \in \mathcal{N}} v_{k,cn}^{i}(t)+R_{kn}^{i,t} \notag\\
&-\sum_{b\in \mathcal{N}} u_{k,nb}^{i}(t)-p_{kn}^{i}(t)e_{kn}(t)]\leq 0, ~ \forall i,k,n.\\
\lim_{T\rightarrow \infty}\frac{1}{T}\sum_{t=1}^T \mathbb{E}&[p_{k'n}^{i}(t)e_{k'n}(t)-\sum_{d\in \mathcal{N}} v_{k,nd}^{i}(t)]\leq 0, \forall i,k,k',n.
\end{align}
\end{subequations}
As a result, \eqref{eq: objective} can be relaxed as
\begin{equation}\label{eq: objective-relax}
\begin{aligned}
\tilde{\Phi}^* &=\min_{\mathcal{X}} \lim_{T\rightarrow \infty}\frac{1}{T}\sum_{t=1}^T \mathbb{E}\{\Phi^t(\mathbf{x}^t)\},\\
&\text{s.t.} \quad \eqref{eq: line rate},\; \eqref{eq: service rate},\;\eqref{eq: Queue-relax},\;\forall t.
\end{aligned}
\end{equation}

Compared to \eqref{eq: objective}, problem \eqref{eq: objective-relax} eliminates the time coupling among variables $\{\mathbf{A}(t),\forall t\}$ by replacing \eqref{eq: Qt},\;\eqref{eq: qt} and \eqref{eq: stable} with \eqref{eq: Queue-relax}. Since \eqref{eq: objective-relax} is a relaxation of \eqref{eq: objective} with its optimal objective $\tilde{\Phi}^*\leq {\Phi}^*$, if one solves \eqref{eq: objective-relax} instead of \eqref{eq: objective}, it is prudent to derive the optimality bound of ${\Phi}^*$, provided that the solution $\mathcal{X}$ for \eqref{eq: objective-relax} is feasible for \eqref{eq: Qt},\;\eqref{eq: qt} and \eqref{eq: stable}, as will be shown in Theorem~1.

We can take a stochastic gradient approach to solving \eqref{eq: objective-relax} with asymptotical optimality guarantee. Concatenate the random parameters into a state vector $\bm{s}^t:=[R_{kn}^{i,t}, \alpha_n^t,\beta_{[a,b]}^t, \forall [a,b], i,k,n]$. For analytic tractability, $\bm{s}^t$ is assumed to be independent and identically distributed (i.i.d.) across slots. (In practice, $\bm{s}^t$ can be non-i.i.d. and even correlated. In such case, the algorithm proposed in this paper can be readily applied. Yet, performance analyses of the non-i.i.d. case can be obtained by generalizing the delayed Lyapunov drift techniques \cite{neely2010stochastic}).
Then, we can rewrite \eqref{eq: objective-relax} as
\begin{subequations}\label{eq: objective-instant}
\begin{align}
\tilde{\Phi}^*& = \min_{\mathcal{X}} ~~~\mathbb{E}\{\Phi^t(\mathcal{X}(\bm{s}^t);\bm{s}^t)\}\\
\text{s.t.} \quad &u_{k,ab}^{i}(\bm{s}^t) \geq 0, \notag\\
&\sum_{i,k} [u_{k,ab}^{i}(\bm{s}^t)+v_{k,ab}^{i}(\bm{s}^t)]=u_{k^*,ab}^{i^*}(\bm{s}^t) (\text{or} \;v_{k^*,ab}^{i^*}(\bm{s}^t))\notag\\
&\leq l_{ab}^{\max},\\
&p_{kn}^{i}(\bm{s}^t) \geq 0, \;\; \sum_{i,k} p_{kn}^{i}(\bm{s}^t) e_{kn}(\bm{s}^t)=p_{k^*n}^{i^*}(\bm{s}^t)\leq p_{n}^{\max},\\
&\mathbb{E} [\sum_{a \in \mathcal{N}} u_{k,an}^{i}(\bm{s}^t)+\sum_{c \in \mathcal{N}} v_{k,cn}^{i}(\bm{s}^t)+R_{kn}^{i,t} \notag\\
&-\sum_{b\in \mathcal{N}} u_{k,nb}^{i}(\bm{s}^t)-p_{kn}^{i}(\bm{s}^t)e_{kn}(\bm{s}^t)]\leq 0, \label{10d}\\
&\mathbb{E}[p_{k'n}^{i}(\bm{s}^t)e_{k'n}(\bm{s}^t)-\sum_{d\in \mathcal{N}} v_{k,nd}^{i}(\bm{s}^t)]\leq 0, \label{10e}
\end{align}
\end{subequations}
where $p_{kn}^{i}(\bm{s}^t):=p_{kn}^{i}(t), e_{kn}(\bm{s}^t)=e_{kn}(t), u_{k,ab}^{i}(\bm{s}^t):=u_{k,ab}^{i}(t), v_{k,ab}^{i}(\bm{s}^t):=v_{k,ab}^{i}(t), ~\forall [a,b],i,k,n$, and $\Phi^t(\mathcal{X}(\bm{s}^t);\bm{s}^t):=\Phi^t(\mathbf{x}^t)$. Formulation \eqref{eq: objective-instant} explicitly indicates the dependence of the decision variables $\{\mathbf{e}^t, \mathbf{p}^t,\mathbf{u}^t,\mathbf{v}^t\}$ on the realization of $\bm{s}^t$.

Let $\mathcal{F}^t$ denote the set of $\{\mathbf{e}^t, \mathbf{p}^t,\mathbf{u}^t,\mathbf{v}^t\}$ satisfying constraints \eqref{eq: line rate} and \eqref{eq: service rate} per $t$, while ${\lambda_{kn,1}^i}$ and ${\lambda_{kn,2}^i}$ denote the Lagrange multipliers associated with the constraints \eqref{10d} and \eqref{10e}. With a convenient notation $\bm{\lambda}:=\{\lambda_{kn,1}^i, \lambda_{kn,2}^i, \forall i,k,n\}$, the partial Lagrangian function of \eqref{eq: objective-instant} is given by
\begin{align}\label{eq.Lam}
L(\mathcal{X},\bm{\lambda}) := \mathbb{E}[L^t(\mathbf{x}^t,\bm{\lambda})]
\end{align}
where the instantaneous Lagrangian is given by \begin{align}\label{eq.instantaneous_Lag}
&L^t(\mathbf{x}^t,\bm{\lambda}):= \Phi^t(\mathbf{x}^t)+ \sum_{i,k,n} \lambda_{kn,1}^i(t)  ( \sum_{a \in \mathcal{N}} u_{k,an}^{i}(t)\notag\\&+\sum_{c \in \mathcal{N}} v_{k,cn}^{i}(t)+R_{kn}^{i,t} -\sum_{b\in \mathcal{N}} u_{k,nb}^{i}(t)-p_{kn}^{i}(t)e_{kn}(t))\notag\\&+ \sum_{i,k,k',n} \lambda_{kn,2}^i(t) (p_{k'n}^{i}(t)e_{k'n}(t)-\sum_{d\in \mathcal{N}} v_{k,nd}^{i}(t)).
\end{align}
Notice that the instantaneous objective $\Phi^t(\mathbf{x}^t)$ and the instantaneous constraints associated with $\bm{\lambda}$ are parameterized by the observed state $\bm{s}^t$ at time $t$; thus the instantaneous Lagrangian can be written as $L^t(\mathbf{x}^t,\bm{\lambda})=L(\mathcal{X}(\bm{s}^t),\bm{\lambda};\bm{s}^t)$, and $L(\mathcal{X},\bm{\lambda})=\mathbb{E}[L(\mathcal{X}(\bm{s}^t),\bm{\lambda};\bm{s}^t)]$.

As a result, the Lagrange dual function is given by
\begin{equation}
D(\bm{\lambda}):=\min_{\{\mathbf{x}^t \in \mathcal{F}^t\}_t}\; L(\mathcal{X},\bm{\lambda}), \label{eq.dual func}
\end{equation}
and the dual problem of \eqref{eq: objective-relax} is: $\max_{\bm{\lambda}\geq 0}\; D(\bm{\lambda})$, where $``\geq"$ is defined entry-wise.

For the dual problem, we can take a standard gradient method to obtain the optimal $\bm{\lambda}^*$ \cite{himmelblau1972applied}.
This amounts to running the following iterations slot by slot
\begin{subequations}\label{eq.subgradient}
\begin{align}
\lambda_{kn,1}^i(t+1) = [\lambda_{kn,1}^i(t)+ \epsilon g_{\lambda_{kn,1}^i}(t)]^+,\, \quad \forall i,k,n, \\
\lambda_{kn,2}^i(t+1) = [\lambda_{kn,2}^i(t)+ \epsilon g_{\lambda_{kn,2}^i}(t)]^+,\, \quad \forall i,k,n.
\end{align}
\end{subequations}
where $\epsilon>0$ is an appropriate stepsize. The gradient $\bm{g}(t):=[g_{\lambda_{kn,1}^i}(t), g_{\lambda_{kn,2}^i}(t), \forall i,k,n]$
can be expressed as
\begin{subequations}\label{eq.subg}
\begin{align}
g_{\lambda_{kn,1}^i}(t)= \mathbb{E}&[\sum_{a \in \mathcal{N}} u_{k,an}^{i}(t)+\sum_{c \in \mathcal{N}} v_{k,cn}^{i}(t)+R_{kn}^{i,t} \notag\\
&-\sum_{b\in \mathcal{N}} u_{k,nb}^{i}(t)-p_{kn}^{i}(t)e_{kn}(t)],\\
g_{\lambda_{kn,2}^i}(t)= \mathbb{E}&[ p_{k'n}^{i}(t)e_{k'n}(t)-\sum_{d\in \mathcal{N}} v_{k,nd}^{i}(t)],
\end{align}
\end{subequations}
where $\mathbf{x}^t:=\{\mathbf{e}^t, \mathbf{p}^t,\mathbf{u}^t,\mathbf{v}^t\}$ is given by
\begin{align}\label{eq.sub}
&\mathbf{x}^t = \argmin_{\mathbf{x}^t \in \mathcal{F}^t} L^t(\mathbf{x}^t,\bm{\lambda}).
\end{align}
%

Note that a challenge associated with \eqref{eq.subg} is sequentially taking expectations over the random vector $\bm{s}^t$ to compute the gradient $\bm{g}(t)$. This would require high-dimensional integration over an unknown probabilistic distribution function of $\bm{s}^t$; or equivalently, computing the corresponding time-averages over an infinite time horizon. Such a requirement is impractical since the computational complexity could be prohibitively high.

To bypass this impasse, we propose to rely on a stochastic dual gradient approach, which is able to combat randomness in the absence of the a-priori knowledge on the statistics of variables. Specifically, dropping $\mathbb{E}$ from \eqref{eq.subg}, we propose the following iterations
\begin{subequations}\label{eq.stoc-subgrad}
\begin{align}
\tilde{\lambda}_{kn,1}^i(t+1) &= \tilde{\lambda}_{kn,1}^i(t)+\epsilon [\sum_{a \in \mathcal{N}} u_{k,an}^{i}(t)+\sum_{c \in \mathcal{N}} v_{k,cn}^{i}(t)\notag\\
&+R_{kn}^{i,t}-\sum_{b\in \mathcal{N}} u_{k,nb}^{i}(t)-p_{kn}^{i}(t)e_{kn}(t)]^+ ,\\
\tilde{\lambda}_{kn,2}^i(t+1) &= \tilde{\lambda}_{kn,2}^i(t)+\epsilon [p_{k'n}^{i}(t)e_{k'n}(t)-\sum_{d\in \mathcal{N}}  v_{k,nd}^{i}(t)]^+,
\end{align}
\end{subequations}
where $\tilde{\bm{\lambda}}^t=\{\tilde{\lambda}_{kn,1}^i(t),\tilde{\lambda}_{kn,2}^i(t), \forall i,k,n\}$ collects the stochastic estimates of the variables in \eqref{eq.subgradient}, and $\mathbf{x}^t(\tilde{\bm{\lambda}})$ is obtained by
solving \eqref{eq.sub} with $\bm{\lambda}$ replaced by $\tilde{\bm{\lambda}}^t$, $\forall i,k,n$.

Note that the interval of updating \eqref{eq.stoc-subgrad} coincides with slots. In other words, the update of \eqref{eq.stoc-subgrad} is an {\em online} approximation of \eqref{eq.subgradient} based on the {\em instantaneous} decisions $\mathbf{x}^t(\tilde{\bm{\lambda}}^t)$ per slot $t$. This stochastic approach becomes possible due to the decoupling of optimization variables over time in \eqref{eq: objective-relax}. 

Relying on the so-called Lyapunov optimization technique in \cite{neely2010stochastic}, we can formally establish that:
\begin{theorem}
If $\bm{s}^t$ is i.i.d. over slots, then the time-average cost of \eqref{eq: objective-instant} with the multipliers updated by \eqref{eq.stoc-subgrad} satisfies
\begin{subequations}
\begin{align}
    \Phi^*\leq  \lim_{T\rightarrow \infty} \frac{1}{T} \sum_{t=0}^{T-1} \mathbb{E}\left[\Phi^t(\mathbf{x}^t))\right] \leq \Phi^* + \epsilon B
\end{align}
where $B=\frac{9}{2}(N^{\max} l^{\max})^2+\frac{3}{2}({R^{\max}}^2+{p^{\max}}^2)$, $N^{\max}$ is the maximum degree of VMs, $l^{\max}=\max_{[a,b]} l_{ab}^{\max}$ and $p^{\max}=\max_{n} p_{n}^{\max}$; $\Phi^*$ is the optimal value of \eqref{eq: objective} under any feasible control policy (i.e., the processing and routing decisions per VM), even if that relies on knowing future realizations of random variables.

Assume that there exists a stationary policy $\mathcal{X}$ and $\mathbb{E} [\sum_{a \in \mathcal{N}} u_{k,an}^{i}(t)+\sum_{c \in \mathcal{N}} v_{k,cn}^{i}(t)+R_{kn}^{i,t}-\sum_{b\in \mathcal{N}} u_{k,nb}^{i}(t)-p_{kn}^{i}(t)e_{kn}(t)]\leq -\zeta$, and $\mathbb{E}[ p_{k'n}^{i}(t)e_{k'n}(t)-\sum_{d\in \mathcal{N}}  v_{k,nd}^{i}(t)]\leq -\zeta$, where $\zeta>0$ is a slack vector constant.
Then all queues are stable, and the time-average queue length satisfies:

\begin{align}
\lim_{T\rightarrow \infty}\frac{1}{T}\sum_{t=1}^T \sum_{i,k,n} \mathbb{E}[ Q_{kn}^{i}(t)+q_{kn}^{i}(t)]= {\cal O}(\frac{1}{\epsilon}).
\end{align}
\end{subequations}
\end{theorem}

\begin{proof}
See Appendices A and B.
\end{proof}

Theorem 1 asserts that the time-average cost of \eqref{eq: objective-instant} obtained by the stochastic dual gradient approach converges to an ${\cal O(\epsilon)}$ neighborhood of the optimal solution, where the region of neighborhood vanishes as the stepsize $\epsilon \rightarrow 0$. The typical tradeoff from the stochastic network optimization holds in this case \cite{neely2010stochastic}: an ${\cal O}(1/\epsilon)$ queue length is necessary, when an ${\cal O(\epsilon)}$ close-to-optimal cost is achieved. Different from \cite{neely2010stochastic}, here the Lagrange dual theory is utilized to simplify the arguments, as shown in Appendices A and B.

\begin{remark}
The asymptotic approximation of the proposed distributed online approach to the cost lower bound achieved offline in a posterior manner is rigorously proved. The lower bound corresponds to the assumption that all the randomnesses are precisely known in prior and the optimal decisions over infinite time-horizon are all derived. This lower bound would violate causality and be computationally prohibitive to achieve, even in an offline fashion, given an infinite number of variables. Theorem 1 indicates that the proposed approach can increasingly approach the lower bound by increasing the tolerance to queue backlogs or delays.
\end{remark}

\subsection{Distributed online implementation}

The dual iteration \eqref{eq.stoc-subgrad} coincides with \eqref{eq: Qt} and \eqref{eq: qt} for $\tilde{\lambda}_{kn,1}^i(t)/\epsilon= Q_{kn}^{i}(t)$ and $\tilde{\lambda}_{kn,2}^i(t)/\epsilon= q_{kn}^{i}(t), \forall i,k,n,t$; this can be interpreted by using the concept of virtual queue of this parallelism \cite{neely2010stochastic}.
With $\tilde{\lambda}_{kn,1}^i(t)$ substituted by $\epsilon Q_{kn}^{i}(t)$ and $\tilde{\lambda}_{kn,2}^i(t)$ substituted by $\epsilon q_{kn}^{i}(t)$, we can obtain the desired $\mathbf{x}^t(\mathbf{A}(t))$ by solving the following problem:
\begin{equation}\label{sub-problem1}
\begin{aligned}
\min_{\mathbf{x}^t \in \mathcal{F}^t} &\frac{1}{\epsilon}\Phi^t(\mathbf{x}^t) + \sum_{i,k,n} Q_{kn}^{i}(t) [ \sum_{a \in \mathcal{N}} u_{k,an}^{i}(t)+\sum_{c \in \mathcal{N}} v_{k,cn}^{i}(t) \\
&+R_{kn}^{i,t}-\sum_{b\in \mathcal{N}} u_{k,nb}^{i}(t)-p_{kn}^{i}(t)e_{kn}(t)] \\
&+\sum_{i,k,k',n} q_{kn}^{i}(t)[p_{k'n}^{i}(t)e_{k'n}^{i}(t)-\sum_{d\in \mathcal{N}} v_{k,nd}^{i}(t)].
\end{aligned}
\end{equation}
Through rearrangement, \eqref{sub-problem1} is equivalent to
\begin{equation}\label{sub-problem2}
\min_{\mathbf{x}^t \in \mathcal{F}^t} ~~\sum_{i,k,n,b\in \mathcal{N}}[f_1(\mathbf{e}^t, \mathbf{p}^t)+f_2(\mathbf{u}^t)+f_3(\mathbf{v}^t)]
\end{equation}
where
{\small \begin{subequations}\label{sub-problem3}
\begin{align}
&f_1(\mathbf{e}^t,\mathbf{p}^t)=  [\frac{\alpha_n^t}{\epsilon}(p_{kn}^{i}(t))^2-(Q_{kn}^{i}(t)-q_{k''n}^{i}(t))p_{kn}^{i}(t)]e_{kn}(t); \label{sub-problem3a}\\
&f_2(\mathbf{u}^t)=[\frac{\beta_{[n,b]}^t}{\epsilon}(u_{k,nb}^{i}(t))^2-(Q_{kn}^{i}(t)-Q_{kb}^{i}(t))u_{k,nb}^{i}(t);\label{sub-problem3b}\\
&f_3(\mathbf{v}^t)=\frac{\beta_{[n,b]}^t}{\epsilon}(v_{k,nb}^{i}(t))^2-(q_{kn}^{i}(t)-Q_{kb}^{i}(t))v_{k,nb}^{i}(t).\label{sub-problem3c}
\end{align}
\end{subequations}}Here, $f_{k''}$ denotes the VNF to be processed after $f_k$ for type-$i$ network services.

Problem \eqref{sub-problem2} can be readily solved by decoupling between $e_{kn}(t)$, $p_{kn}^{i}(t)$, $u_{k,nb}^{i}(t)$ and $v_{k,nb}^{i}(t)$, and between the VMs.
Specifically, \eqref{sub-problem2} can be decoupled into the following subproblems per VM or per inter-VM link:
\begin{subequations}\label{sub-problem4}
\begin{align}
&\min_{\mathbf{e}^t, \mathbf{p}^t}\quad f_1(\mathbf{e}^t; \mathbf{p}^t),\label{sub-problem4a}\\
& \min_{\mathbf{u}^t}\quad f_2(\mathbf{u}^t);\label{sub-problem4b}\\
&\min_{\mathbf{v}^t}\quad f_3(\mathbf{v}^t)\label{sub-problem4c}.
\end{align}
\end{subequations}

Problem \eqref{sub-problem4a} is a mixed integer programming. Its solution can be obtained by comparing the minimums of $f_1(\mathbf{e}^t, \mathbf{p}^t)$ separately achieved under $e_{kn}(t)=0$ and 1. In the case of $e_{kn}(t)=0$, $f_1(\mathbf{e}^t, \mathbf{p}^t)=0$. In the case of $e_{kn}(t)=1$, \eqref{sub-problem4a} becomes the minimization of a quadratic function of $p_{kn}^{i}(t)$, where the optimal solution is given by
\begin{equation}\label{p_kn}
~~~{p_{kn}^{i}}^*(t) =\min\{\max\{ \frac{\epsilon(Q_{kn}^{i}(t)-q_{k''n}^{i}(t))}{2\alpha_n^t}, 0\},p_n^{\max}\}, ~\forall i,k.
\end{equation}
Then, the optimal objective of \eqref{sub-problem4a} can be obtained by substituting \eqref{p_kn} into $f_1(\mathbf{e}^t, \mathbf{p}^t)$, as given by
\begin{equation}\label{weights1}
~~P_{kn}^i:=
    \begin{cases}
       -\frac{\epsilon(Q_{kn}^{i}(t)-q_{k''n}^{i}(t))^2}{4\alpha_n^t}, ~~\text{if} ~ Q_{kn}^{i}(t)-q_{k''n}^{i}(t)>0; \\
       \quad \quad \quad 0,  \quad \quad \quad\quad\quad~\text{if} ~Q_{kn}^{i}(t)-q_{k''n}^{i}(t)\leq 0. \\
    \end{cases}
 \end{equation}
Since every VM only runs a single VNF (i.e., $\sum_k e_{kn}(t)=1, \forall n,t$), we have
\begin{equation}\label{e_kn}
{e_{kn}}^*(t)=
    \begin{cases}
       1, \quad \text{if} ~~k=\arg \min_k P_{kn}^i; \\
       0, \quad \text{otherwise}. \\
    \end{cases}
    \end{equation}

Problems \eqref{sub-problem4b} and \eqref{sub-problem4c} are the minimizations of quadratic functions of $\mathbf{u}^t$ and $\mathbf{v}^t$, respectively. Like \eqref{sub-problem4a} under $e_{kn}(t)=1$, the optimal solutions for \eqref{sub-problem4b} and \eqref{sub-problem4c} are
\begin{equation}\label{eq.first order2}
{u_{k,nb}^{i}}^*(t) =\min\{\max\{  \frac{\epsilon(Q_{kn}^{i}(t)-Q_{kb}^{i}(t))}{2\beta_{[n,b]}^t}, 0\},l_{ab}^{\max}\},~\forall i,k;\nonumber
\end{equation}
\begin{equation}
{v_{k,nb}^{i}}^*(t) =\min\{\max\{\frac{\epsilon(q_{kn}^{i}(t)-Q_{kb}^{i}(t))}{2\beta_{[n,b]}^t},0\},l_{ab}^{\max}\},~\forall i,k;
\end{equation}
with their corresponding objectives given by

\begin{equation}
U_{k,nb}^i:=
    \begin{cases}
       -\frac{\epsilon(Q_{kn}^{i}(t)-Q_{kb}^{i}(t))^2}{4\beta_{[n,b]}^t} , ~~\text{if} ~ Q_{kn}^{i}(t)-Q_{kb}^{i}(t)>0; \\
       \quad \quad \quad 0,  \quad \quad \quad\quad~~~\text{if} ~Q_{kn}^{i}(t)-Q_{kb}^{i}(t)\leq 0; \\
    \end{cases}\nonumber
\end{equation}
\begin{equation}
V_{k,nb}^i:=
    \begin{cases}
       -\frac{\epsilon(q_{kn}^{i}(t)-Q_{kb}^{i}(t))^2}{4\beta_{[n,b]}^t} , ~~\text{if} ~ q_{kn}^{i}(t)-Q_{kb}^{i}(t)>0; \\
       \quad \quad \quad 0,  \quad \quad \quad\quad~~\text{if} ~q_{kn}^{i}(t)-Q_{kb}^{i}(t)\leq 0. \\
    \end{cases}\nonumber
\end{equation}
\begin{equation}\label{weights2}
\end{equation}

Recall that any VM $n$ or directional link $[a,b]$ can only process or transmit a single network service per slot.
At each slot, a VM can prioritize the queues of different service types to be processed by different VNFs, and process or route services from the queue with the highest priority.
The priority is ranked based on the objectives $P_{kn}^i$, $U_{k,nb}^i$ and $V_{k,nb}^i$ in \eqref{weights1} and \eqref{weights2}.
For this reason, we refer to $P_{kn}^i$, $U_{k,nb}^i$ and $V_{k,nb}^i$ as queue-price objectives. The processing and routing decisions can be made by one-to-one mapping between the queues and outgoing links/processor to minimize the total of selected non-zero objectives, as summarized in Algorithm 1.

\begin{algorithm}[t]
\caption{Distributed Online Optimization of NFV}\label{algo}
\begin{algorithmic}[1]
\For {$t=1,2\dots$}
\State Each VM $n$ observes the queue lengths of its own and its one-hop neighbors.
\State Install VNF $f_k$ at VM $n$ based on \eqref{e_kn}.
\State Repeatedly send network services to the VM processor or outgoing links with the minimum non-zero
       queue-price objectives in \eqref{weights1} and \eqref{weights2}, using the optimal rates derived in \eqref{p_kn} and \eqref{eq.first order2}, until either the processor and all outgoing links are scheduled or the remaining objectives are all zero.
\State Update $Q_{kn}^{i}(t)$ and $q_{kn}^{i}(t)$ for all nodes and services via the dynamics \eqref{eq: Qt} and \eqref{eq: qt}.
\EndFor
\end{algorithmic}
\end{algorithm}

Note that Algorithm 1 is decentralized, since every VM only needs to know the queue lengths of its own and its immediate neighbors. 
Optimal decisions of a VM, locally made by comparing the queue-price objectives, comply with (17) and therefore preserve the asymptotic optimality of the entire network, as dictated in Theorem 1. With decentralized decision makings, Algorithm 1 can readily provide improved flexibility and scalability, alleviate signaling burden, and reduce service latency for practical NFV systems.

\begin{figure}[t]
\centering
\includegraphics[height=0.34\textwidth]{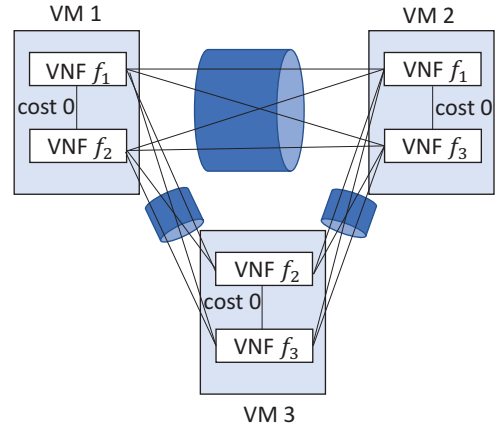}
\caption{An illustration on VMs running multiple VNFs, where VNFs can be interpreted as ``VMs'' and VMs can be interpreted as ``clusters of VMs.'' Then the VM-based online optimization developed in this paper can be readily applied to the ``VMs.'' }
\label{fig: nfv2}
\end{figure}

Also note that Algorithm~1 can be readily extended to general scenarios where a VM runs multiple VNFs; see Fig.~\ref{fig: nfv2}. In this case, all VNFs can be first interpreted as separate ``VMs'' in the context of the baseline case of a VNF per VM, and colocated VNFs at a VM then become a cluster of multiple ``VMs.'' No cost incurs on the connections between the ``VMs'' within the cluster.
The only difference from the baseline scenario of a VNF per VM, as described in Algorithm~1, is that, between the clusters, only a pair of ``VMs'' which are respectively from the two clusters and the most cost-effective to transmit workloads, can be activated.
This can be achieved by comparing the price weights of the links to pick up the most cost-effective link between clusters.
The optimal decisions of processing at each of the ``VMs'' stay unchanged.



\section{Optimal Placement and Operation of NFV at Different Timescales}

In this section, we consider a more practical scenario where the placement of VNFs is carried out at the VMs at a much larger time interval, i.e., at time $\tau= mT_\Delta (m=1,2,\ldots$), rather than on a slot basis.
This is because the installation of VNFs at the VMs could cause interruptions to network service provisions. We prove that if the placement and the operation of NFV are jointly optimized at two different timescales, the aforementioned asymptotic optimality of the proposed approach can be preserved.

\subsection{Two-timescale placement and operation}
By evaluating \eqref{sub-problem4} at two different timescales, the placement of VNFs, and the processing and routing of network services, can be carried out as follows:
\begin{itemize}
\item \textit{Placement of VNFs at a $T_\Delta$-slot interval:} At time slot $\tau= mT_\Delta$, each VM $n$ decides on the VNF to install to minimize the expectation of the sum of $f_1(\mathbf{e}^t, \mathbf{p}^t)$ in \eqref{sub-problem4a} over the time window $t=\{\tau, \ldots, \tau+T_\Delta-1\}$, i.e., $\mathbb{E}\big\{\sum_{t=\tau}^{\tau+T_\Delta-1} f_1(\mathbf{e}^t, \mathbf{p}^t)\big\}$, as given by

\begin{align}\label{twoscale1}
\min_{\mathbf{e}^t}&~~\mathbb{E}\Big\{\sum_{t=\tau}^{\tau+T_\Delta-1}\sum_{i,k}[\frac{\alpha_n^t}{\epsilon}(p_{kn}^{i}(t))^2-(Q_{kn}^{i}(t)\notag
\\&-q_{k''n}^{i}(t)) p_{kn}^{i}(t)]e_{kn}(\tau)\Big\}.
\end{align}

\item \textit{Processing and routing of network services per slot $t$:} Per slot $t$, each VM processes and routes network services, following Algorithm 1, given the placement decisions of VNFs given by \eqref{twoscale1}.

\item \textit{Queue update:} Each VM updates its queues $Q_{kn}^{i}(t)$ and $q_{kn}^{i}(t)$ at every slot $t$ based on \eqref{eq: Qt} and \eqref{eq: qt}.
\end{itemize}

Note that the optimal solutions to \eqref{twoscale1} require future knowledge on service arrivals $\{R_{kn}^{i,t}, t=\tau,\ldots, \tau+T_\Delta-1\}$, and the prices of service processing and routing $\{\alpha_n^t,\beta_{[a,b]}^t, t=\tau,\ldots, \tau+T_\Delta-1\}$. This would violate causality. We propose to take an approximation by setting the future queue backlogs as their current backlogs at slot $\tau = mT$, as given by
\begin{subequations}\label{approximation}
\begin{align}
\hat{Q}_{kn}^{i}(t) &= Q_{kn}^{i}(\tau);\\
\hat{q}_{kn}^{i}(t) &= q_{kn}^{i}(\tau),~\forall t = \tau, \ldots, \tau+T_\Delta-1,
\end{align}
\end{subequations}
where $\hat{Q}_{kn}^{i}(t)$ and $\hat{q}_{kn}^{i}(t)$ are the approximated queue backlogs. Taking the approximation of \eqref{approximation} and that the stochastic variables $\bm{s}^t:=[R_{kn}^{i,t}, \alpha_n^t,\beta_{[a,b]}^t, \forall [a,b], i,k,n]$ to be invariant in the coming time window, \eqref{twoscale1} can be reduced to the per-slot problem, as given in \eqref{sub-problem4a}. Therefore, as per time slot $\tau = mT_\Delta$, the VNFs can be installed at the VMs based on \eqref{e_kn}, as summarized in Algorithm~1.

\subsection{Optimality loss of VNF placement}
We next analyze the optimality loss due to the approximation of \eqref{approximation}. We can prove that the optimality loss is bounded and does not compromise the asymptotic optimality of the proposed approach. To improve tractability, an upper bound of the optimality loss is evaluated in the case where all the variables $\{\mathbf{e}^t, \mathbf{p}^t,\mathbf{u}^t,\mathbf{v}^t\}$ are optimized under the assumption of the availability of the full knowledge on the future $T_\Delta$ time slots, rather than taking the approximation of \eqref{approximation}.

Based on \eqref{approximation}, we first establish the following lemma:
\begin{lemma}
At any slot $t$, the differences between the approximated and actual queue backlogs in \eqref{approximation} are bounded by
\begin{subequations}\label{bounded queue}
\begin{align}
\big|Q_{kn}^{i}(t)-\hat{Q}_{kn}^{i}(t)\big | &\leq T_\Delta \omega_Q,\label{bounded queue1}\\
\big|q_{kn}^{i}(t)-\hat{q}_{kn}^{i}(t) \big| &\leq T_\Delta \omega_q,\label{bounded queue2}
\end{align}
\end{subequations}
where the constants $\omega_Q = \max\{N^{\max}l^{\max}+p^{\max}, 2N^{\max}l^{\max}+R^{\max}\}$, and $\omega_q = \max\{N^{\max}l^{\max}, p^{\max}\}$.
\end{lemma}
\begin{proof}
For any two consecutive slots $t$ and $t+1$, the difference of the queue backlogs is bounded, i.e., $\big|Q_{kn}^{i}(t+1)-Q_{kn}^{i}(t)\big| \leq \omega_Q$, where $\omega_Q$ is the maximum difference between the departure and the arrival of $\mathbf{Q}(t)$, denoted by $\max\{N^{\max}l^{\max}+p^{\max}, 2N^{\max}l^{\max}+R^{\max}\}$. According to \eqref{approximation} and the inequality $|a+b|\leq |a|+|b|$, we have
$\big|Q_{kn}^{i}(t)-\hat{Q}_{kn}^{i}(t)\big|=\big|Q_{kn}^{i}(t)-Q_{kn}^{i}(\tau)\big|=\big|\sum_{t_0=\tau}^t [Q(t_0+1)-Q(t_0)]\big| \leq (t-\tau)\omega_Q \leq T_\Delta\omega_Q$, where $\tau=mT_\Delta$ and $t=\tau, \ldots, \tau+T_\Delta-1$. Therefore, \eqref{bounded queue1} is proved. Likewise, \eqref{bounded queue2} can be proved.
\end{proof}

Based on Lemma 1, we are ready to obtain the following theorem:
\begin{theorem}
The optimality loss of the solution for \eqref{sub-problem2} due to the approximation of the queue backlogs in \eqref{approximation} is bounded, i.e.,
\begin{align}\label{bounded theorem}
&\Big|\sum_{i,k,n,b\in \mathcal{N}}[f_1(\hat{\mathbf{e}}^t, \hat{\mathbf{p}}^t)+f_2(\hat{\mathbf{u}}^t)+f_3(\hat{\mathbf{v}}^t)]\notag\\
&-\sum_{i,k,n,b\in \mathcal{N}}[f_1(\mathbf{e}^t, \mathbf{p}^t)+f_2(\mathbf{u}^t)+f_3(\mathbf{v}^t)]\Big| \leq \epsilon C,
\end{align}
where $C:=\epsilon {T_\Delta}^2 N^2 |\mathcal{I}|K[(\frac{1}{2\alpha^{\max}}+\frac{1}{2\beta^{\max}})(\omega_Q+\omega_q)^2+\frac{2}{\beta^{\max}}\omega_Q^2]$; $\alpha^{\max}=\max_{n,t} \alpha_n^t$, and $\beta^{\max}=\max_{[a,b],t}\beta_{[a,b]}^t$.
\end{theorem}
\begin{proof}
See Appendix C.
\end{proof}
We can see from Theorem 2 that the optimality loss of the two-timescale control increases quadratically with the time interval of VNF placement, $T_\Delta$. This allows us to balance between the optimality loss and the cost of VNF placement. As dictated in Theorems 1 and 2, the total optimality loss of the two-timescale approach for problem \eqref{eq: objective} is upper bounded, as given by
\begin{align}
    \overline{\Phi(\hat{\mathbf{x}}^t)} \leq \Phi^* + \epsilon (B+C),
\end{align}
where $\overline{\Phi(\hat{\mathbf{x}}^t)}$ is the time-average cost under the two-timescale approach. In other words, the two-timescale placement and operation of NFV preserves the asymptotic optimality with approximated queue backlogs.
%

\subsection{Learn-and-adapt for placement acceleration}

While proved to preserve the asymptotic optimality, the larger
timescale can slow down the optimal placement of VNFs. We propose to speed the placement up through a learning and adaptation method~\cite{chen2017learn}. The Lagrange multipliers $\tilde{\lambda}_{kn,1}^i(t)$ and $\tilde{\lambda}_{kn,2}^i(t)$ play the key roles in the proposed distributed online optimization of NFV in \eqref{eq.stoc-subgrad}. We can incrementally learn these Lagrange multipliers from observed data and speed up the convergence of the multipliers driven by the learning process.

\begin{algorithm}[t]
\caption{Distributed Learn-and-Adapt NFV Optimization}\label{algo1}
\begin{algorithmic}[1]
\For {$t=1,2\dots$}
\State \textbf{Online processing and routing (1st gradient):}
\State Construct the effective dual variable via \eqref{eq.dual-effect}, observe the current state $\bm{s}^t$, and obtain placement, processing and routing decisions $\mathbf{x}^t(\bm{\gamma}^t)$ by minimizing online Lagrangian \eqref{eq.real-time1}.
\State Update the instantaneous queue length $\textbf{Q}(t+1)$ and $\mathbf{q}(t+1)$ with $\mathbf{x}^t(\bm{\gamma}^t)$ via queue dynamics \eqref{eq: Qt} and \eqref{eq: qt}.
\State \textbf{Statistical learning (2nd gradient):} \State Obtain
variable $\mathbf{x}^t(\hat{\bm{\lambda}}^t)$ by solving online
Lagrangian minimization with sample $\bm{s}^t$ via \eqref{eq.real-time2}.
\State Update the empirical dual variable $\hat{\bm{\lambda}}^{t+1}$ via \eqref{eq.dual-lambda-alg}.
\EndFor
\end{algorithmic}
\end{algorithm}


In the proposed learn-and-adapt scheme, with the online learning of $\tilde{\lambda}_{kn,1}^i(t)$ and $\tilde{\lambda}_{kn,2}^i(t), \forall n,i$ at each slot $t$, two stochastic gradients are updated using the current $\bm{s}^t$. 
The first gradient $\bm{\gamma}^{t}$ is designed to minimize the instantaneous Lagrangian for optimal decision makings on processing or routing network services, as given by [cf. \eqref{eq.sub}]
\begin{equation}\label{eq.real-time1}
   \mathbf{x}^t(\bm{\gamma}^t)=\arg\min_{\mathbf{x}^t \in \mathcal{F}^t}{    L}^t(\mathbf{x}^t,\bm{\gamma}^t)
\end{equation}
which depends on what we term \textit{effective} multiplier $\bm{\gamma}^t:=\{{\gamma}_{kn,1}^{i}(t),{\gamma}_{kn,2}^{i}(t), \forall n,i\}$, as given by
\begin{equation}\label{eq.dual-effect}
    \!\underbrace{~~~~\bm{\gamma}^t~~~~}_{\rm effective~multiplier}=\underbrace{~~~~\hat{\bm{\lambda}}^t~~~~}_{\rm statistical~learning}+~~\underbrace{~~~\epsilon \mathbf{A}(t)~-~\bm{\theta}~~~}_{\rm online~adaptation},
\end{equation}
where $\hat{\bm{\lambda}}^t:=\{\hat{\lambda}_{kn,1}^i(t), \hat{\lambda}_{kn,2}^i(t), \forall i,k,n\}$ is the empirical dual variable, and $\bm{\theta}$ controls the bias of $\bm{\gamma}^t$ in the steady state, and can be judiciously designed to achieve the improved cost-delay tradeoff, as will be shown in Theorem 3.

For a better illustration of the effective multiplier in \eqref{eq.dual-effect}, we call $\hat{\bm{\lambda}}(t)$ the statistically learnt dual variable to obtain the exact optimal argument of the dual problem $\max_{\bm{\lambda}\succeq 0}\; D(\bm{\lambda})$. We call $\epsilon\mathbf{A}(t)$ (which is exactly $\bm{\lambda}$ as shown in \eqref{eq.stoc-subgrad}) the online adaptation term, since it can track the instantaneous change of system statistics. The control variable $\epsilon$ tunes the weights of these two factors.

The second gradient is designed to simply learn the stochastic gradient of \eqref{eq.dual func} at the previous empirical dual variable $\hat{\bm{\lambda}}^t$, and implement a gradient ascent update as
 \begin{align}\label{eq.dual-lambda-alg}
\hat{\lambda}_{kn,1}^i(t+1) &= \hat{\lambda}_{kn,1}^i(t)+\eta(t) [\sum_{a \in \mathcal{N}} u_{k,an}^{i}(\hat{\lambda}_{kn,1}^i(t))+R_{kn}^{i,t} \notag\\
&+\sum_{c \in \mathcal{N}} v_{k,cn}^{i}(\hat{\lambda}_{kn,1}^i(t))
-\sum_{b\in \mathcal{N}} u_{k,nb}^{i}(\hat{\lambda}_{kn,1}^i(t))\notag\\
&-p_{kn}^{i}(\hat{\lambda}_{kn,1}^i(t))e_{kn}(\hat{\lambda}_{kn,1}^i(t))]^+ \notag\\
\hat{\lambda}_{kn,2}^i(t+1) &= \hat{\lambda}_{kn,2}^i(t)+\eta(t) [p_{k'n}^{i}(\hat{\lambda}_{k'n,2}^i(t))e_{k'n}(\hat{\lambda}_{k'n,1}^i(t))\notag\\
&-\sum_{d\in \mathcal{N}}  v_{k,nd}^{i}(\hat{\lambda}_{kn,2}^i(t))]^+
\end{align}
where $\eta(t)$ is a proper diminishing stepsize, and the ``virtual'' allocation $\mathbf{x}^t(\hat{\bm{\lambda}}^t)$ can be found by solving
\begin{equation}\label{eq.real-time2}
	\mathbf{x}^t(\hat{\bm{\lambda}}^t) = \arg\min_{\mathbf{x}^t \in \mathcal{F}^t}{
L}^t(\mathbf{x}^t,\hat{\bm{\lambda}}^t).
\end{equation}

With learn-and-adaption incorporated, Algorithm 2 takes an additional learning step in Algorithm 1, i.e., \eqref{eq.dual-lambda-alg}, which adopts gradient ascent with diminishing stepsize $\eta(t)$ to find the ``best empirical'' dual variable from all observed network states. In the transient stage, the extra gradient evaluations and empirical dual variables accelerate the convergence speed of Algorithm 1; while in the steady stage, the empirical dual variable approaches the optimal multiplier, which significantly reduces the steady-state queue lengths.



%

Using the learn-and-adapt approach, we are ready to arrive at the following theorem \cite[Theorems 2 and 3]{chen2017learn}.
\begin{theorem}\label{the.queue-stable}
Suppose that the assumptions in Theorem~1 are satisfied. Then with $\bm{\gamma}^t$ defined in \eqref{eq.dual-effect} and $\bm{\theta}={\cal
O}(\sqrt{\epsilon}\log^2(\epsilon))$, Algorithm 2 yields a near-optimal solution for \eqref{eq: objective} in the sense that
\begin{equation}\label{eq.opt-gap}
        {\Phi}^*\leq\lim_{T\rightarrow \infty} \frac{1}{T} \sum_{t=1}^{T} \mathbb{E}\left[\Phi^t\left(\mathbf{x}^t(\bm{\gamma}^t)\right)\right] \leq {\Phi}^*+{\cal O}(\epsilon).\;
\end{equation}
The long-term average expected queue length satisfies
\begin{equation}\label{eq.stt-length}
    \lim_{T\rightarrow \infty}\frac{1}{T}\sum_{t=1}^T \sum_{i,k,n} \mathbb{E}[ Q_{kn}^{i}(t)+q_{kn}^{i}(t)]={\cal O}\left(\frac{\log^2(\epsilon)}{\sqrt{\epsilon}}\right),~
\end{equation}
where $\mathbf{x}^t(\bm{\gamma}^t)$ denotes the real-time operations obtained from the Lagrangian minimization
\eqref{eq.real-time1}.
\end{theorem}

Theorem \ref{the.queue-stable} asserts that by setting $\bm{\theta}={\cal O}(\sqrt{\epsilon}\log^2(\epsilon))$,
Algorithm 2 is asymptotically ${\cal O}(\epsilon)$-optimal with an average queue
length ${\cal O}(\log^2(\epsilon)/{\sqrt{\epsilon}})$. This implies
that the algorithm is able to achieve a near-optimal cost-delay tradeoff $[\epsilon,\log^2(\epsilon)/{\sqrt{\epsilon}}]$; see \cite{chen2017learn}.
Comparing with the standard tradeoff
$[\epsilon,{1}/{\epsilon}]$ under Algorithm 1, the learn-and-adapt design of
Algorithm 2 remarkably improves the delay performance.




%

\section{Numerical Tests}

\begin{figure}[t]

\centering
\includegraphics[height=0.36\textwidth]{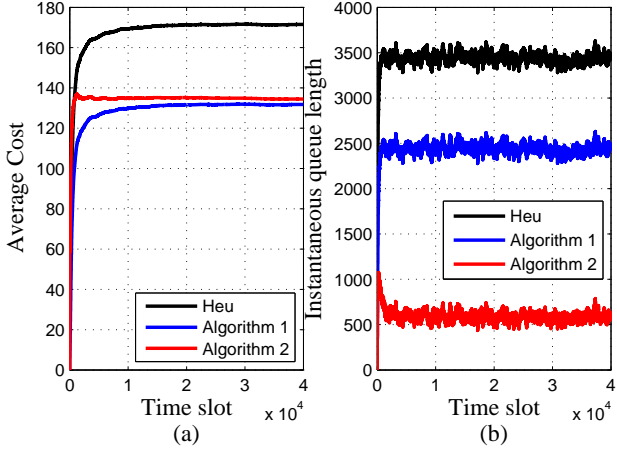}

\caption{Comparison of time-average costs and instantaneous queue lengths, where $N=7$, $\epsilon=0.1$ and average arrival rate is 14 services/sec.}

\label{fig: cost-t}
\end{figure}

\begin{figure}[t]
\centering
\includegraphics[height=0.36\textwidth]{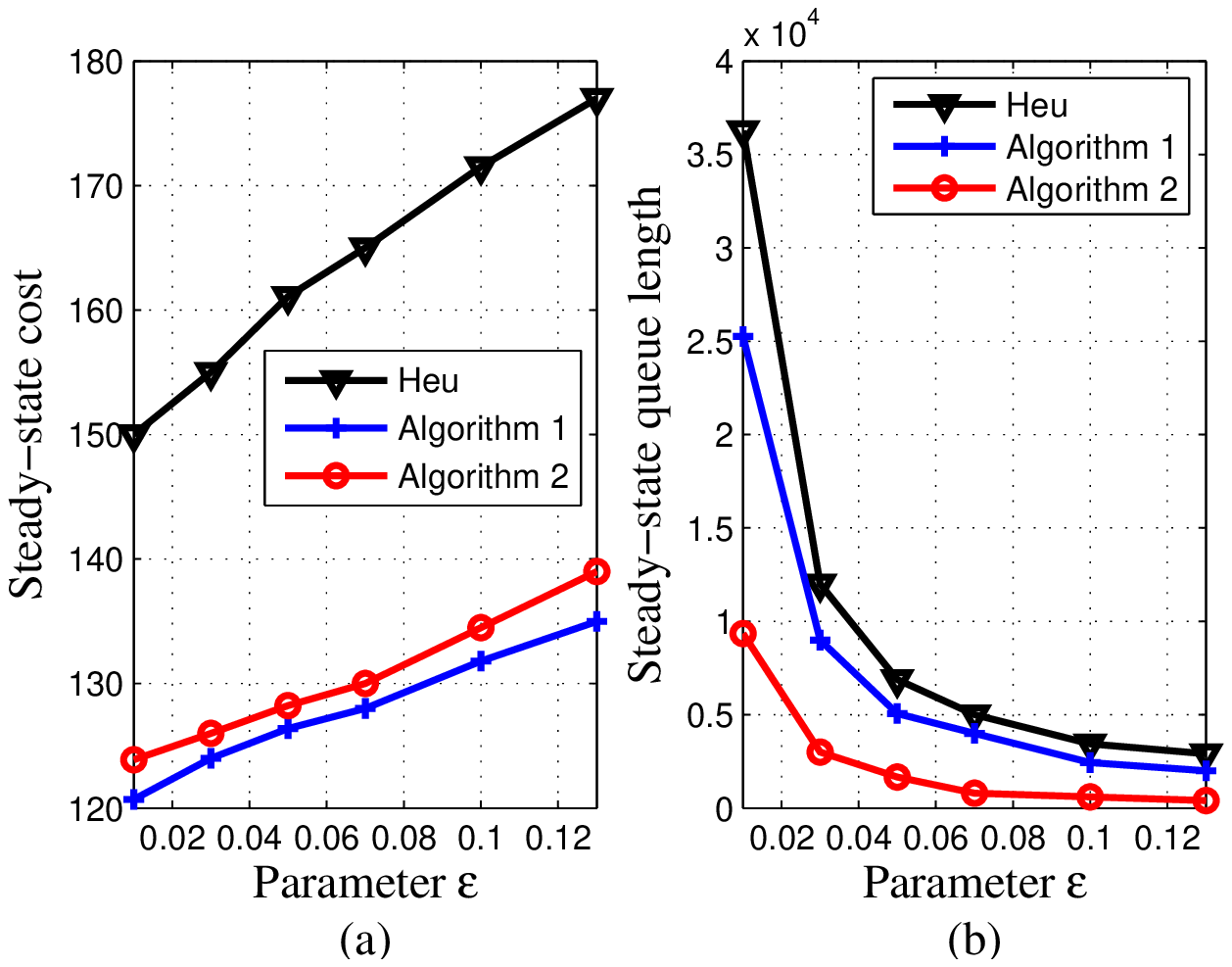}

\caption{Comparison of steady-state costs and queue lengths under different $\epsilon$, where $N=7$ and average arrival rate is 14 services/sec.}

\label{fig: cost-mu}
\end{figure}

\begin{figure}[t]
\centering
\includegraphics[height=0.36\textwidth]{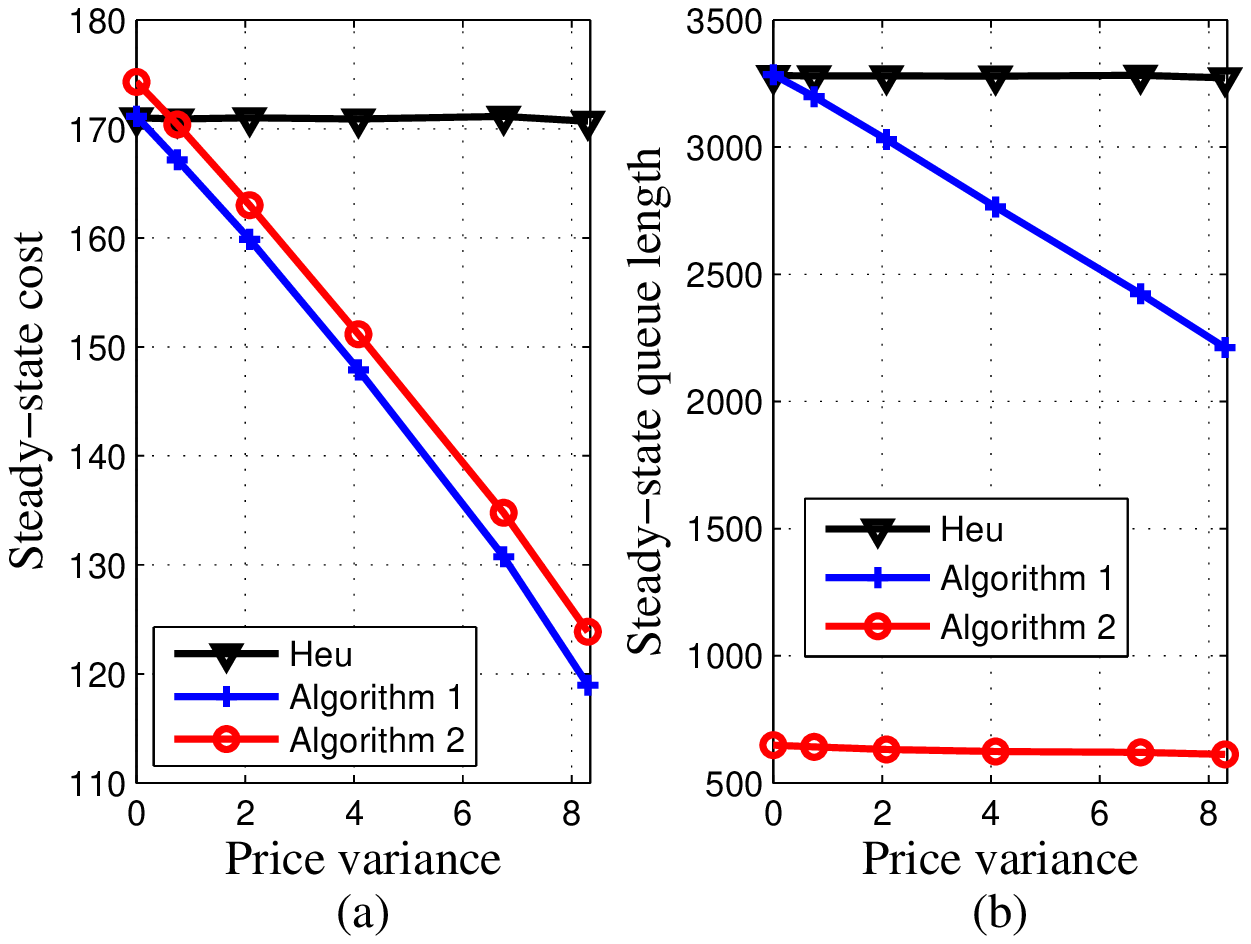}

\caption{Comparison of steady-state costs and queue lengths under different price variances, where $N=7$, $\epsilon=0.1$ and average arrival rate is 14 services/sec.}

\label{fig: cost-queue-pricenew}
\end{figure}

\begin{figure}[t]
\centering
\includegraphics[height=0.35\textwidth]{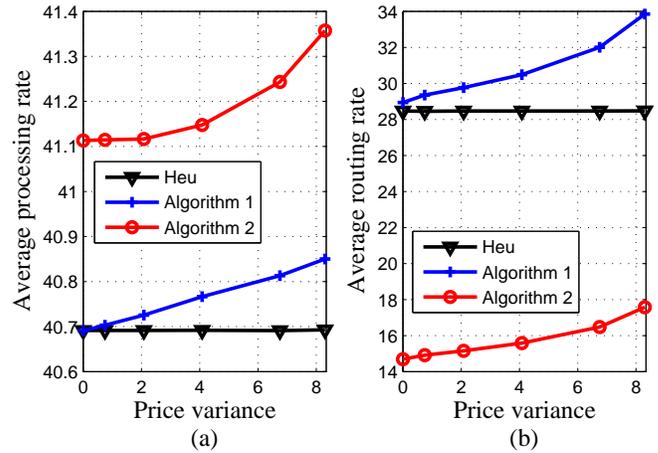}

\caption{Comparison of average processing and routing rates for all network services on all VMs under different price variances, where $N=7$, $\epsilon=0.1$ and average arrival rate is 14 services/sec.}

\label{fig: rate-price}
\end{figure}

\begin{figure}[t]
\centering
\includegraphics[height=0.4\textwidth]{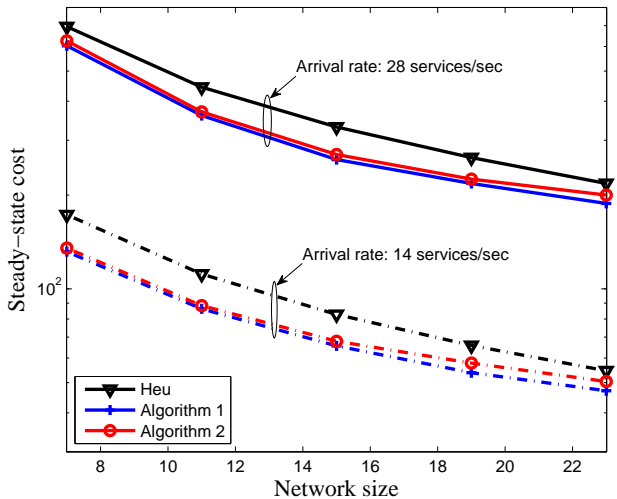}

\caption{Comparison of steady-state costs under different network sizes, where $\epsilon=0.1$.}

\label{fig: cost-size}
\end{figure}

\begin{figure}[t]
\centering
\includegraphics[height=0.4\textwidth]{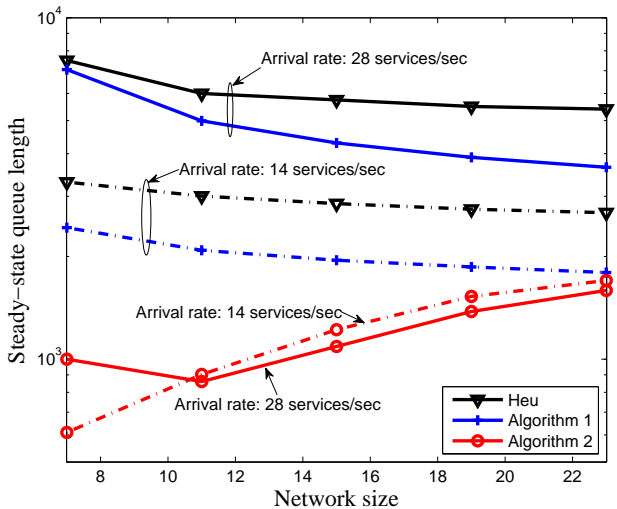}

\caption{Comparison of steady-state queue lengths under different network sizes, where $\epsilon=0.1$.}

\label{fig: queue-size}
\end{figure}

Numerical tests are provided to validate our analytical claims and demonstrate the merits of the proposed algorithms. Two types of network services are considered on the platform with $N=7$ VMs. The first type of network  service is $\{f_1, f_2, f_3\}$ and the second type of network service is $\{f_3, f_1, f_2\}$.
Suppose that each service has a size of 1 KB.
The processing and routing prices $\alpha_n^t$ and $\beta_{[a,b]}^t$ are uniformly distributed over $[0.1,1]$ by default; $p_{n}^{\max}$ and $l_{ab}^{\max}$ are generated from a uniform distribution within $[10,20]$. The default arrival rate of network services is uniformly distributed with a mean of 14 services/sec.
The stepsize is $\eta(t)=1/\sqrt{t},\forall t$, the tradeoff variable is $\epsilon=0.1$, and the bias correction vector is chosen as $\bm{\theta}=2\sqrt{\epsilon}\log^2(\epsilon)$.
Algorithms are evaluated in a two-timescale scenario, where the placement of VNFs is carried out every $T_\Delta=5$ sec. In addition to the proposed Algorithms~1 and 2, we also simulate a heuristic algorithm (Heu) as the benchmark, which decides the placement of VNFs and processing/routing rates similarly as Algorithm~1, but with the prices in \eqref{p_kn} and \eqref{eq.first order2} replaced by their means. Therefore, the decisions are made only based on queue differences, with no price considerations.

Fig.~\ref{fig: cost-t} compares the three algorithms in terms of the time-average cost and the instantaneous queue length. It can be seen from Fig.~\ref{fig: cost-t}(a) that the time-average cost of Algorithm 2 converges slightly higher than that of Algorithm 1, while the time-average cost of Heu is about 30\% larger.
Furthermore, Algorithm 2 exhibits faster convergence than Algorithm 1 and Heu, as its time-average cost quickly reaches the optimal steady-state value by leveraging the learning process.
Fig.~\ref{fig: cost-t}(b) shows that Algorithm 2 incurs the shortest queue lengths among the three algorithms, followed by Algorithm 1. Particularly, the aggregated instantaneous queue length of Algorithm 2 is about 76\% and 83\% smaller than those of Algorithm 1 and Heu, respectively.
Clearly, the learn-and-adapt procedure reduces delay without markedly compromising the time-average cost.

Fig.~\ref{fig: cost-mu} compares the steady-state cost and queue length of the three algorithms, under different stepsize (tradeoff coefficient) $\epsilon$.
It is observed that as $\epsilon$ grows, the steady-state costs of all three algorithms increase and the steady-state queue lengths declines. This validates our findings in Theorems 1 and 3. 

The steady-state cost and queue length are also compared under different price variances in Figs.~\ref{fig: cost-queue-pricenew}(a) and (b). Here, processing and routing prices are generated with the mean of 0.55 and variance from $3.3\times 10^{-5}$ to $8.3\times 10^{-2}$. The costs and queue lengths of Algorithms 1 and 2 decrease as the price variance increases, while those of Heu remain unchanged. This is because Heu adopts price-independent processing and routing rates, while Algorithms 1 and 2 are able to minimize the cost by taking advantage of price differences among VMs and links. 
As further shown in Figs.~\ref{fig: rate-price}(a) and (b), the average processing and routing rates of Algorithms 1 and 2 rise with the growth of price variance, since the algorithms either choose a lower priced link with a higher routing rate, or a lower priced VM with a higher processing rate.

An interesting finding is that the average backlog of Algorithm 2 is insusceptible to price variances; see Fig.~\ref{fig: cost-queue-pricenew}(b). This is due to the fact that the algorithm, aiming to reduce the backlog of unfinished network services, achieves the aim by avoiding routing network services within the same type of VMs. This can also be evident from Figs.~\ref{fig: rate-price}(a) and (b), where VNFs are processed typically at the first encountered corresponding VMs, even at higher processing rates, hence reducing routing rates.

Fig.~\ref{fig: cost-size} plots the steady-state costs of Algorithms 1 and 2, and Heu, as the network size $N$ (i.e., the number of VMs) increases. It can be observed in Fig.~\ref{fig: cost-size} that under the same arrival rate of services, the costs decline as the network becomes large. This is due to the increased connectivity of each VM, which helps increase the diversity of choosing cost-effective routing links and neighboring VMs with low prices and, in turn, reduce the costs. The costs increase as the average arrival rate of services increases, since more resources are required to accommodate the increased arrivals.

In Fig.~\ref{fig: queue-size}, we plot the steady-state queue lengths of Algorithms~1 and 2, and Heu, as the network size grows. We can see that Algorithms~1 and 2 are able to reduce the queue length of the network, as compared to Heu. The reduction of Algorithm~1 is increasingly large, especially when the arrival rate of services is large. In addition, the queue length of Algorithm~2 under the arrival rate of 28 services/sec can become lower than that under the arrival rate of 14 services/sec, as the network becomes large. We can conclude that the gain of Algorithm~2 diminishes, as the network size grows with a relatively light arrival rate of services. Nevertheless, Algorithm~2 is more efficient under heavier traffic arrivals in a large network.

\section{Conclusions}

In this paper, a new distributed online optimization was developed to minimize the time-average cost of NFV, while stabilizing the function queues of VMs. Asymptotically optimal decisions on the placement of VNFs, and the processing and routing of network services were instantly generated at individual VMs, adapting to the topology and stochasticity of
the network. A learn-and-adapt approach was further proposed to speed up stabilizing the VMs and achieve a cost-delay tradeoff $[\epsilon,\log^2(\epsilon)/{\sqrt{\epsilon}}]$. Numerical results show that the proposed method is able to reduce the time-average cost of NFV by 30\% and reduce the queue length by 83\%.

\appendix
\subsection{Proof of (18a) in Theorem 1}
\begin{proof}
From the recursions in \eqref{eq: Qt}, we have
\begin{align}
  & (Q_{kn}^i(t+1))^2
    =[Q_{kn}^i(t) + \sum_{a \in \mathcal{N}} u_{k,an}^{i}(t)+\sum_{c \in \mathcal{N}} v_{k,cn}^{i}(t)\notag\\
    &+R_{kn}^{i,t}-\sum_{b\in \mathcal{N}} u_{k,nb}^{i}(t)-p_{kn}^{i}(t)e_{kn}(t)]^2 \notag\\
    & \leq (Q_{kn}^i(t))^2 + 2Q_{kn}^i(t)[\sum_{a \in \mathcal{N}} u_{k,an}^{i}(t)+\sum_{c \in \mathcal{N}} v_{k,cn}^{i}(t)\notag\\&+R_{kn}^{i,t}-\sum_{b\in \mathcal{N}} u_{k,nb}^{i}(t) -p_{kn}^{i}(t)e_{kn}(t)]\notag\\&+ \underbrace{8(N^{\max} l^{\max})^2+3{R^{\max}}^2+2{p^{\max}}^2}_{2B_1} \notag
\end{align}
where $N^{\max}$ is the maximum degree of VMs, $l^{\max}=\max_{[a,b]} l_{ab}^{\max}$ and $p^{\max}=\max_{n} p_{n}^{\max}$.
Similarly, we also have
\begin{align}
  & (q_{kn}^i(t+1))^2 =[q_{kn}^i(t) + p_{k'n}^{i}(t)e_{k'n}(t)-\sum_{d\in \mathcal{N}}  v_{k,nd}^{i}(t)]^2 \notag\\
    & \leq (q_{kn}^i(t))^2 + 2q_{kn}^i(t)[p_{k'n}^{i}(t)e_{k'n}(t)-\sum_{d\in \mathcal{N}} \notag v_{k,nd}^{i}(t)]\\&+ \underbrace{(N^{\max} {l^{\max}})^2+{p^{\max}}^2}_{2B_2}. \notag
\end{align}

Considering now the Lyapunov function $\Upsilon(t):=\frac{1}{2} [\sum_{i,k,n} (Q_{kn}^i(t))^2+\sum_{i,k,n} (q_{kn}^i(t))^2]$, it readily follows that
\begin{align}
    & \triangle \Upsilon(t):=\Upsilon(t+1)-\Upsilon(t) \notag \\
    & \leq \sum_{i,k,n} \{Q_{kn}^i(t)[\sum_{a \in \mathcal{N}} u_{k,an}^{i}(t)+\sum_{c \in \mathcal{N}} v_{k,cn}^{i}(t)+R_{kn}^{i,t}  \notag\\
    &-\sum_{b\in \mathcal{N}} u_{k,nb}^{i}(t)-p_{kn}^{i}(t)e_{kn}(t)]\}\notag \\
    &+\sum_{i,k,n} \{q_{kn}^i(t)[p_{k'n}^{i}(t)e_{k'n}(t)-\sum_{d\in \mathcal{N}} v_{k,nd}^{i}(t)]\} + B, \notag
\end{align}where $B := B_1+B_2$ is a constant.
Taking expectations and adding $\frac{1}{\epsilon}\mathbb{E}[\Phi^t(\mathbf{x}^t)]$ ($\mathbf{x}^t$ is the optimal policy by solving \eqref{eq.sub}) to both sides, we arrive at
\begin{align}
    & \mathbb{E}[\triangle \Upsilon(t)] + \frac{1}{\epsilon}\mathbb{E}[\Phi^t(\mathbf{x}^t)] \notag \\
    & \leq B + \frac{1}{\epsilon} \mathbb{E}\Bigl(\Phi^t(\mathbf{x}^t) + \epsilon \sum_{n,i}[ Q_{kn}^{i}(t) ( \sum_{a \in \mathcal{N}} u_{k,an}^{i}(t))\notag  \\
    &+\sum_{c \in \mathcal{N}} v_{k,cn}^{i}(t)
    +R_{kn}^{i,t}-\sum_{b\in \mathcal{N}} u_{k,nb}^{i}(t)-p_{kn}^{i}(t))e_{kn}(t)] \notag \\
    &+\epsilon \sum_{n,i}[ q_{kn}^{i}(t)(p_{k'n}^{i}(t)e_{k'n}(t)-\sum_{d\in \mathcal{N}} v_{k,nd}^{i}(t))] \Bigr) \notag \\
    & = B + \frac{1}{\epsilon} L(\mathcal{X}(\epsilon \mathbf{A}(t)), \epsilon \mathbf{A}(t)) \notag \\
    & =  B + \frac{1}{\epsilon} D(\epsilon \mathbf{A}(t))  \leq B + \frac{1}{\epsilon} \tilde{\Phi}^*\notag
\end{align}where we use the definition of $L(\mathcal{X}, \bm{\lambda})$ in (\ref{eq.Lam}); $\mathcal{X}(\epsilon \mathbf{A}(t))$ denotes the optimal primal variable set given by~\eqref{eq.sub} for $\bm{\lambda} = \epsilon \mathbf{A}(t)$ (hence, $L(\mathcal{X}(\epsilon \mathbf{A}(t)), \epsilon \mathbf{A}(t)) = D(\epsilon \mathbf{A}(t))$); $\tilde{\Phi}^*$ denotes the optimal value for problem \eqref{eq: objective-relax}; and the last inequality is due to the weak duality: $D(\bm{\lambda}) \leq \tilde{\Phi}^*$, $\forall \bm{\lambda}$.

Summing over all $t$, we then have
\begin{align}
    & \sum_{t=0}^{T-1} \mathbb{E}[\triangle \Upsilon(t)] + \frac{1}{\epsilon} \sum_{t=0}^{T-1} \mathbb{E}[\Phi^t(\mathbf{x}^t)] \notag \\
    & = \mathbb{E}[\Upsilon(T)]-\Upsilon(0) + \frac{1}{\epsilon} \sum_{t=0}^{T-1} \mathbb{E}[\Phi^t(\mathbf{x}^t)]  \leq T (B + \frac{1}{\epsilon} \tilde{\Phi}^*) \notag
\end{align}which leads to
\begin{align}
    \frac{1}{T} \sum_{t=0}^{T-1} \mathbb{E}[\Phi^t(\mathbf{x}^t)] & \leq \tilde{\Phi}^* + \epsilon (B + \frac{\Upsilon(0)}{T}) \leq \Phi^* + \epsilon (B + \frac{\Upsilon(0)}{T}). \notag
\end{align}(18a) follows by taking $T \rightarrow \infty$.
\end{proof}

\subsection{Proof of (18b) in Theorem 1}

Assume that there exists a stationary policy $\mathcal{X}$, under which $\mathbb{E} [\sum_{a \in \mathcal{N}} u_{k,an}^{i}(t)+\sum_{c \in \mathcal{N}} v_{k,cn}^{i}(t)+R_{kn}^{i,t}-\sum_{b\in \mathcal{N}} u_{k,nb}^{i}(t)-p_{kn}^{i}(t)e_{kn}(t)]\leq -\zeta$, and $\mathbb{E}[ p_{k'n}^{i}(t)e_{k'n}(t)-\sum_{d\in \mathcal{N}} v_{k,nd}^{i}(t)]\leq -\zeta$, where $\zeta>0$ is a slack vector constant, we have the following lemma. 

\begin{lemma}
If the random state $\bm{s}^t$ is i.i.d., there exists a \textit{stationary} control policy $\mathcal{P}^{\text{stat}}$, which is a pure (possibly randomized) function of the realization of $\bm{s}^t$, satisfying \eqref{eq: line rate} and \eqref{eq: service rate}, and providing the following guarantees per $t$:
\begin{align}\label{stationary_policy}
&\mathbb{E}[\Phi^{\text{stat}}(\mathbf{x}^t)]=\tilde{\Phi}^*, \notag\\
&\mathbb{E} [\sum_{a \in \mathcal{N}} u_{k,an}^{i,\text{stat}}(t)+\sum_{c \in \mathcal{N}} v_{k,cn}^{i,\text{stat}}(t)+R_{kn}^{i,t} \notag\\
&-\sum_{b\in \mathcal{N}} u_{k,nb}^{i,\text{stat}}(t)-p_{kn}^{i,\text{stat}}(t)e_{kn}^{\text{stat}}(t)]\leq -\zeta, \notag\\
&\mathbb{E}[p_{k'n}^{i,\text{stat}}(t)e_{k'n}^{\text{stat}}(t)-\sum_{d\in \mathcal{N}} v_{k,nd}^{i,\text{stat}}(t)]\leq -\zeta,
\end{align}where $\Phi^{\text{stat}}(\mathbf{x}^t)$ denotes the resultant cost, $\{e_{kn}^{\text{stat}}(t),p_{kn}^{i,\text{stat}}(t),u_{k,ab}^{i,\text{stat}}(t), v_{k,ab}^{i,\text{stat}}(t), \forall [a,b], i,k,n \}$ denote the routing and processing rates under policy $\mathcal{P}^{\text{stat}}$, and expectations are taken over the randomization of $\bm{s}^t$ and (possibly) $\mathcal{P}^{\text{stat}}$.

\end{lemma}

\begin{proof}
The proof argument is similar to that in \cite[Theorem 4.5]{neely2010stochastic}; hence, it is omitted for brevity.
\end{proof}

It is worth noting that \eqref{stationary_policy} not only assures that the stationary control policy $\mathcal{P}^{\text{stat}}$ achieves the optimal cost for \eqref{eq: objective-relax}, but also guarantees that the resultant expected cost per slot $t$ is equal to the optimal time-averaged cost (due to the stationarity of $\bm{s}^t$ and $\mathcal{P}^{\text{stat}}$).

Now from (27) we have
\begin{align}
    & \mathbb{E}[\triangle \Upsilon(t)] + \frac{1}{\epsilon}\mathbb{E}[\Phi^t(\mathbf{x}^t)] \notag \\
    & \leq B + \frac{1}{\epsilon} \mathbb{E}\Bigl(\Phi^{\text{stat}}(\mathbf{x}^t) + \epsilon \sum_{n,i}[ Q_{kn}^{i}(t) ( \sum_{a \in \mathcal{N}} u_{k,an}^{i,\text{stat}}(t))\notag  \\
    &+\sum_{c \in \mathcal{N}} v_{k,cn}^{i,\text{stat}}(t)
    +R_{kn}^{i,t}-\sum_{b\in \mathcal{N}} u_{k,nb}^{i,\text{stat}}(t)-p_{kn}^{i,\text{stat}}(t))e_{kn}^{\text{stat}}(t))] \notag \\
    &+\epsilon \sum_{i,k,n}[ q_{kn}^{i}(t)(p_{k'n}^{i,\text{stat}}(t)e_{k'n}^{\text{stat}}(t))-\sum_{d\in \mathcal{N}} v_{k,nd}^{i,\text{stat}}(t))] \Bigr) \notag \\
    & \leq B + \frac{1}{\epsilon} \Phi^*- \zeta \sum_{i,k,n}\mathbb{E}[ Q_{kn}^{i}(t)+q_{kn}^{i}(t)],
\end{align}where the first equality holds since Algorithm 1 minimizes the instantaneous Lagrangian $L^t$ in \eqref{eq.instantaneous_Lag} among all policies satisfying \eqref{eq: line rate} and \eqref{eq: service rate}, including $\mathcal{P}^{\text{stat}}$; and the last inequality is due to Lemma~1.

Summing over all $t$, we then have
\begin{align}
    & \sum_{t=0}^{T-1} \mathbb{E}[\triangle \Upsilon(t)] + \frac{1}{\epsilon} \sum_{t=0}^{T-1} \mathbb{E}[\Phi^t(\mathbf{x}^t)] \notag \\
    & = \mathbb{E}[\Upsilon(T)]-\Upsilon(0) + \frac{1}{\epsilon} \sum_{t=0}^{T-1} \mathbb{E}[\Phi^t(\mathbf{x}^t)] \notag \\
    & \leq T(B + \frac{\Phi^*}{\epsilon})- \zeta \sum_{t=0}^{T-1} \sum_{i,k,n}\mathbb{E}[ Q_{kn}^{i}(t)+q_{kn}^{i}(t)] \notag
\end{align}which leads to
\begin{align}
&\frac{1}{T}\sum_{t=0}^{T-1} \sum_{i,k,n}\mathbb{E}[ Q_{kn}^{i}(t)+q_{kn}^{i}(t)]\leq \frac{1}{\zeta}(B + \frac{\Phi^*}{\epsilon})+\frac{\Upsilon(0)}{\zeta T}.
\end{align}(18b) follows by taking $T \rightarrow \infty$.

\subsection{Proof of Theorem 2}
From \eqref{sub-problem3a}, we can get
\begin{align}\label{app1}
&f_1(\hat{\mathbf{e}}^t,\hat{\mathbf{p}}^t|\mathbf{A})\notag\\
&=  [\frac{\alpha_n^t}{\epsilon}(\hat{p}_{kn}^{i}(t))^2-(Q_{kn}^{i}(t)-q_{k''n}^{i}(t))\hat{p}_{kn}^{i}(t)]\hat{e}_{kn}(t) \notag \\
&= [\frac{\alpha_n^t}{\epsilon}(\hat{p}_{kn}^{i}(t))^2-(\hat{Q}_{kn}^{i}(t)-\hat{q}_{k''n}^{i}(t))\hat{p}_{kn}^{i}(t)]\hat{e}_{kn}(t) \notag\\
&+[(\hat{Q}_{kn}^{i}(t)-Q_{kn}^{i}(t))-(\hat{q}_{k''n}^{i}(t)-q_{k''n}^{i}(t))]\hat{p}_{kn}^{i}(t)\hat{e}_{kn}(t).
\end{align}

Recall that $\{\hat{\mathbf{e}}^t,\hat{\mathbf{p}}^t\}$ is the optimal solution under the approximate $\hat{\mathbf{A}}$, we have
\begin{align}\label{app2}
&f_1(\hat{\mathbf{e}}^t,\hat{\mathbf{p}}^t|\hat{\mathbf{A}})\notag\\
&=[\frac{\alpha_n^t}{\epsilon}(\hat{p}_{kn}^{i}(t))^2-(\hat{Q}_{kn}^{i}(t)-\hat{q}_{k''n}^{i}(t))\hat{p}_{kn}^{i}(t)]\hat{e}_{kn}(t) \notag \\
&\leq f_1(\mathbf{e}^t,\mathbf{p}^t|\hat{\mathbf{A}})\notag\\
&=[\frac{\alpha_n^t}{\epsilon}(p_{kn}^{i}(t))^2-(\hat{Q}_{kn}^{i}(t)-\hat{q}_{k''n}^{i}(t))p_{kn}^{i}(t)]e_{kn}(t) \notag \\
&=[\frac{\alpha_n^t}{\epsilon}(p_{kn}^{i}(t))^2-(Q_{kn}^{i}(t)-q_{k''n}^{i}(t))p_{kn}^{i}(t)]e_{kn}(t) \notag\\
&+[(Q_{kn}^{i}(t)-\hat{Q}_{kn}^{i}(t))-(q_{k''n}^{i}(t)-\hat{q}_{k''n}^{i}(t))]p_{kn}^{i}(t)e_{kn}(t). \notag\\
&=f_1(\mathbf{e}^t,\mathbf{p}^t|\mathbf{A})\notag\\
&+[(Q_{kn}^{i}(t)-\hat{Q}_{kn}^{i}(t))-(q_{k''n}^{i}(t)-\hat{q}_{k''n}^{i}(t))]p_{kn}^{i}(t)e_{kn}(t).
\end{align}

Combining \eqref{app1} and \eqref{app2}, we have
\begin{subequations}\label{difference}
\begin{align}\label{difference1}
&|f_1(\hat{\mathbf{e}}^t,\hat{\mathbf{p}}^t|\mathbf{A})-f_1(\mathbf{e}^t,\mathbf{p}^t|\mathbf{A})| \leq
|[(\hat{Q}_{kn}^{i}(t)-Q_{kn}^{i}(t))\notag\\
&-(\hat{q}_{k''n}^{i}(t)-q_{k''n}^{i}(t))][\hat{p}_{kn}^{i}(t)\hat{e}_{kn}(t)-p_{kn}^{i}(t)e_{kn}(t)]| \notag\\
&\leq[|\hat{Q}_{kn}^{i}(t)-Q_{kn}^{i}(t)|+|\hat{q}_{k''n}^{i}(t)-q_{k''n}^{i}(t)|]\cdot \notag\\
&\frac{\epsilon}{2\alpha_n^t}[|\hat{Q}_{kn}^{i}(t)-Q_{kn}^{i}(t)|+|\hat{q}_{k''n}^{i}(t)-q_{k''n}^{i}(t)|]\notag\\
&\leq\frac{\epsilon}{2\alpha_n^t}(T_\Delta\omega_Q+T_\Delta\omega_q)^2.
\end{align}
The second inequality holds due to \eqref{p_kn} and the inequality $|a-b|\leq |a|+|b|$.
Likewise, we can get the optimality loss of \eqref{sub-problem3b} and \eqref{sub-problem3c}, as given by
\begin{align}
&|f_2(\hat{\mathbf{u}}^t)-f_2(\mathbf{u}^t)| \leq \frac{2\epsilon}{\beta_{[n,b]}^t}(T_\Delta\omega_Q)^2, \label{difference2}\\
&|f_3(\hat{\mathbf{v}}^t)-f_3(\mathbf{v}^t)| \leq \frac{\epsilon}{2\beta_{[n,b]}^t}(T_\Delta\omega_Q+T_\Delta\omega_q)^2. \label{difference3}
\end{align}
\end{subequations}
Adding up \eqref{difference}, we prove the theorem.


\end{document}